\documentclass[nohyperref]{article}
\usepackage[preprint]{neurips_2021}
\usepackage{microtype}
\usepackage{graphicx}
\usepackage{subfigure}
\usepackage{booktabs} 
\usepackage{hyperref}

\usepackage{amsmath}
\usepackage{amssymb}
\usepackage{mathtools}
\usepackage{amsthm}
\usepackage{adjustbox}
\usepackage{multirow}
\usepackage{babel}
\usepackage{float}

\usepackage{verbatim}

\usepackage[capitalize,noabbrev]{cleveref}

\theoremstyle{plain}
\newtheorem{theorem}{Theorem}[section]
\newtheorem{proposition}[theorem]{Proposition}

\theoremstyle{definition}

\theoremstyle{remark}

\usepackage[textsize=tiny]{todonotes}

\newcommand{\tr}[1]{{#1}^\top}
\newcommand{\vect}[1]{\mathbf{#1}}

\newcommand{\res}[2]{#1}

\newcommand{\proj}{g}
\newcommand{\tsf}{\mathcal{T}}

\newcommand{\pp}{\vect{p}}

\usepackage{bbm}
\newcommand{\vone}{\mathbbm{1}}

\newcommand{\KL}{D_{\mr{KL}}}

\newcommand{\dec}{f_D}
\newcommand{\enc}{f_E}

\DeclareMathOperator{\Expect}{\mathbb{E}}
\newcommand{\entr}{\mathcal{H}}
\DeclareMathOperator{\entropy}{\entr}
\newcommand{\MI}{\mathcal{I}}

\newcommand{\Real}{\mathbb{R}}
\newcommand{\img}{\mathbf{x}}
\newcommand{\data}{\mathcal{D}}

\newcommand{\mr}[1]{\text{#1}}
\newcommand{\loss}[1]{\mathcal{L}_{\mr{#1}}}

\newcommand{
 \tabincell}[2]{\begin{tabular}{@{}#1@{}}
 #2
\end{tabular}
}

\newcommand{\Pjoint}{\mathcal{P}_{\mr{joint}}}
\newcommand{\Ppseudo}{\mathcal{P}_{\mr{pseud}}}

\newcommand{\embed}{\mathbf{s}}
\newcommand{\pHat}{\widehat{\mathbf{p}}}
\newcommand{\pTilde}{\widetilde{\mathbf{p}}}

\newcommand{\mHat}{\widehat{\mathbf{p}}}
\newcommand{\mTilde}{\widetilde{\mathbf{p}}}

\newcommand{\zz}{\mathbf{z}}
\newcommand{\uu}{\mathbf{u}}

\newcommand{\ident}{\mr{I}}

\newcommand{\phz}{\phantom{0}}

\newcommand{\visCluster}[2]{
\includegraphics[width=.124\textwidth]{boundary-aware-figure/new_clusters/#2/case#1/image.png} &
\includegraphics[width=.124\textwidth]{boundary-aware-figure/new_clusters/#2/case#1/gt.png} &
\includegraphics[width=.124\textwidth]{boundary-aware-figure/new_clusters/#2/case#1/superpixel.png} &
\includegraphics[width=.124\textwidth]{boundary-aware-figure/new_clusters/#2/case#1/imsatpp.png} &
\includegraphics[width=.124\textwidth]{boundary-aware-figure/new_clusters/#2/case#1/iic.png} &
\includegraphics[width=.124\textwidth]{boundary-aware-figure/new_clusters/#2/case#1/alpha_0.png} &
\includegraphics[width=.124\textwidth]{boundary-aware-figure/new_clusters/#2/case#1/alpha_025.png} &
\includegraphics[width=.124\textwidth]{boundary-aware-figure/new_clusters/#2/case#1/alpha_050.png} &
\includegraphics[width=.124\textwidth]{boundary-aware-figure/new_clusters/#2/case#1/alpha_075.png} &
\includegraphics[width=.124\textwidth]{boundary-aware-figure/new_clusters/#2/case#1/alpha_1.png} 
}

\newcommand{\visinpsect}[2]{
\includegraphics[width=.124\textwidth]{boundary-aware-figure/seg_inspection/#1/case#2/gt.png} &
\includegraphics[width=.124\textwidth]{boundary-aware-figure/seg_inspection/#1/case#2/ps.png} &
\includegraphics[width=.124\textwidth]{boundary-aware-figure/seg_inspection/#1/case#2/mt.png} &
\includegraphics[width=.124\textwidth]{boundary-aware-figure/seg_inspection/#1/case#2/adv.png} &
\includegraphics[width=.124\textwidth]{boundary-aware-figure/seg_inspection/#1/case#2/contrast-pre.png} &
\includegraphics[width=.124\textwidth]{boundary-aware-figure/seg_inspection/#1/case#2/ours-pre.png} &
\includegraphics[width=.124\textwidth]{boundary-aware-figure/seg_inspection/#1/case#2/contrast-mt.png} &
\includegraphics[width=.124\textwidth]{boundary-aware-figure/seg_inspection/#1/case#2/ours-mt.png}
}

\newcommand*{\correspondingauthor}{\thanks{Corresponding author}}

\title{Boundary-aware Information Maximization for  Self-supervised Medical Image Segmentation}
\author{Jizong Peng\correspondingauthor  \\
  ETS Montreal \\
  \texttt{jizong.peng.1@etsmtl.net} \\
 \And
  Ping Wang \\ 
   ETS Montreal \\
  \texttt{ping.wang.1@ens.etsmtl.ca} \\
 \And
 Christian Desrosiers \\ 
   ETS Montreal \\
  \texttt{christian.desrosiers@etsmtl.ca} \\
 \And
 Marco Pedersoli \\
   ETS Montreal \\
  \texttt{marco.pedersoli@etsmtl.ca} \\
 }
\begin{document}
    \maketitle

	
	
\title{Diversified Multi-prototype Representation for Semi-supervised Segmentation}
\author{Jizong Peng\correspondingauthor  \\
  ETS Montreal \\
  \texttt{jizong.peng.1@etsmtl.net} \\
 \And
 Christian Desrosiers \\ 
   ETS Montreal \\
  \texttt{christian.desrosiers@etsmtl.ca} \\
 \And
 Marco Pedersoli \\
   ETS Montreal \\
  \texttt{marco.pedersoli@etsmtl.ca} \\
 }




\begin{abstract}
	Unsupervised pre-training has been proven as an effective approach to boost various downstream tasks given limited labeled data. Among various methods, contrastive learning learns a discriminative representation by constructing positive and negative pairs. However, it is not trivial to build reasonable pairs for a segmentation task in an unsupervised way. In this work, we propose a novel unsupervised pre-training framework that avoids the drawback of contrastive learning. Our framework consists of two principles: unsupervised over-segmentation as a pre-train-task using Mutual information maximization and boundary-aware preserving learning. Experimental results on two benchmark medical segmentation datasets reveal our method's effectiveness in improving segmentation performance when few annotated images are available.
\end{abstract}

\section{Introduction}

Supervised deep learning approaches have achieved outstanding performance in a wide range of segmentation tasks~\citep{ronneberger2015unet,badrinarayanan17segnet,chen18deeplab}. However, these approaches often require a large amount of labeled images which are difficult to obtain for medical imaging applications~\citep{cheplygina2019not,peng2021medical}. Unsupervised representation learning~\citep{jing2020self,liu2021self} has emerged as an effective technique to boost the performance of a segmentation model without the need for annotated data. In such technique, a model is pre-trained to perform a given pretext task, for example puzzle-solving~\citep{noroozi2016unsupervised,taleb2021multimodal}, rotation prediction~\citep{komodakis2018unsupervised,gidaris2018unsupervised}, colorization~\citep{zhang2016colorful} or contrastive-based instance discrimination ~\citep{hjelm2018learning,chen2020simple,he2020momentum}, and then fine-tuned with a small set of labeled examples. Among these self-supervised methods, contrastive learning has become a prevailing  strategy for pre-training medical image segmentation models~\citep{chaitanya2020contrastive,zeng2021positional,peng2021self}. The core idea of this strategy is to learn, without pixel-wise annotations, an image representation which can discriminate related images (e.g., two transformations of the same image) from non-related ones. 
Most contrastive learning approaches for segmentation apply a contrastive loss on the \emph{global} representation of images, which typically corresponds to the features produced by the network's encoder. Experimental results have shown that pre-training the encoder with this loss and then fine-tuning the whole network with few labeled examples can lead to significant improvements~\citep{peng2021self}.

Recent works have also demonstrated the benefit of using contrastive learning on the decoder's feature maps during pre-training~\citep{chaitanya2020contrastive,peng2021self}. In this case, the contrastive loss is applied at each position of the feature map, which helps learn a \emph{local} representation of the image. 
However, choosing the pairs of positive and negative examples that need to be contrasted is more challenging for these dense feature maps without pixel-wise annotations. Firstly, the meta-information in medical data (e.g., subject ID, slice position, etc.) is typically found at the image level, and is therefore not applicable to local contrastive learning. To tackle this problem, current methods usually adopt a stride sampling strategy where, for a given anchor position in the feature map, local representations located at a sufficient distance are regarded as negative, while those that are close but obtained under different image transforms are considered as positive~\citep{chaitanya2020contrastive}. As we show in our experiments (see Section \ref{sec:results1}), this weak spatial prior unfortunately leads to low improvements when used in pre-training. Another problem comes from the fact that medical images for segmentation are often dominated by non-informative background regions, which reduces the effectiveness of local contrastive learning in this setting. 
Additionally, standard contrastive learning techniques such as~\citep{hjelm2018learning} typically need large batch sizes to have a sufficient amount of high-quality negative example pairs. This constraint can be hard to meet in the case of learning dense features. Despite important efforts, the improvement brought by \emph{local} contrastive learning in medical image segmentation remains relatively marginal  \citep{chaitanya2020contrastive}.
 
In this paper, we propose a boundary-aware information maximization approach for unsupervised representation learning and experimentally demonstrate its usefulness for medical image segmentation. Our approach focuses on the dense features in the decoder of a segmentation network, and seeks to group them into clusters that correspond to meaningful regions in the image. The proposed learning objective is based on the Information Invariant Clustering (IIC) method~\citep{ji2019invariant}, but overcomes three major drawbacks of this method: i) its optimization difficulty, caused in part by minimizing the entropy of cluster assignments,  which often leads to sub-optimal solutions; ii) its lack of clustering consistency for different random transformations; iii) the poor correspondence of clusters obtained by this method with region boundaries in the image. As illustrated in Fig.~\ref{fig:alpha_explain}, our boundary-aware information maximization approach learns clusters that better correspond to relevant anatomical structures of the image. This is achieved by improving IIC in two important ways. First, we augment the learning objective of IIC, which maximizes the mutual information of local feature embeddings for two different transformations of the same image, to make the joint cluster probability close to a uniform diagonal matrix. This improves optimization and leads to clusters that are well balanced and also consistent across different image transformations. Second, we propose a boundary-aware loss based on the cross-correlation between the spatial entropy of clusters and image edges, which helps the learned cluster be more representative of important regions in the image. Our experimental results reveal this loss to be especially effective for the segmentation of regions with irregular shape. 

Compared to contrastive learning, our method does not require to compute  positive or negative pairs, and does not need a sophisticated sampling mechanism or large batch sizes. Through an extensive set of experiments involving four different medical image segmentation tasks, we demonstrate the high effectiveness of our unsupervised representation learning method for pre-training a segmentation model, before fine-tuning it with few labeled images. Our results show the proposed method to outperform by a large margin several state-of-the-art self-supervised and semi-supervised approaches for segmentation, and to reach a performance close to full supervision with only a few labeled examples.

\section{The proposed method}

In unsupervised representation pre-training, we are given a set of $N$ images $\data = \{ \img_i \}_{i=1}^{N}$, with $\img_i \in \mathbb{R}^\Omega$, where $\Omega$ is the image space. We seek to learn a useful representation by pre-training a deep segmentation network $f_\theta(\cdot) = \dec(\enc(\cdot))$ comprised of encoder $\enc(\cdot)$ and a decoder $\dec(\cdot)$. In our setting, a good representation can boost segmentation performance when fine-tuning the whole network with very limited labeled data. To help understand our method, we summarize in Table~\ref{tab:defintions} the main notations used in the paper. 
\begin{table}[t!]
	\footnotesize
	\centering
	\caption{Notations used in the paper}
	\begin{adjustbox}{max width=0.75\linewidth}
		{\renewcommand{\arraystretch}{1.2}
			\begin{tabular}{ll}
\toprule
Image dataset: & $\data=\{\img_i \in \Real^{\Omega}\}_{i=1}^{N}$\\ 
Pixel index: & $\Omega = [1,\dots, W\times H]$ \\
Dense embedding index: & $\Omega' = [1,\dots, W/N\times H/N]$\\
$(K\!-\!1)$-simplex: & $\Delta^K = \{\pp \in [0, 1]^K, \sum\limits_k p_k =1 \}$ \\ 
Dense embedding: & $\embed = s(\img) \in \Real^{\Omega'\times C}$\\ 
Cluster projection: & $\proj(\embed) \in \Delta^{\Omega' \times K}$\\
\midrule
Image transform: & $\tsf(\cdot)$ \\
Cluster probabilities: & $\pHat_i = \proj(s(\tsf(\img_i)))$, \,$\pTilde_i = \proj(\tsf(s(\img_i)))$\\ 
Cluster marginals: & $\mHat = \frac{1}{N} \sum_{i=1}^{N} \pHat_i$, \,$\mTilde = \frac{1}{N} \sum_{i=1}^{N} \pTilde_i $ \\
Joint distribution: & $\Pjoint =\frac{1}{N} \sum_{i=1}^{N} \pHat_i \cdot \pTilde_i^{\intercal}$ \\
\midrule
Entropy: & $\entr(X) = -\Expect_{X} [\log p(X)]$ \\
Joint entropy: & $\entr_{\text{joint}}(X,Y) = -\Expect_{X,Y} [ \log p(X,Y)]$ \\
\bottomrule
\end{tabular}
}
\label{tab:defintions}
\end{adjustbox}
\end{table}

Our representation operates on dense embeddings taken from some intermediate layer of the decoder. Our goal is to group these local embeddings into clusters reflecting meaningful anatomical structures in the input images, without requiring any labels. Three separate loss functions are used to achieve this goal. The first loss maximizes the MI between corresponding local feature embeddings obtained from an input image and its transformed version. Since computing MI between continuous variables is complex, as in recent works \citep{ji2019invariant,peng2021boosting}, we project features to a discrete space representing clusters, where MI is easy to obtain. However, maximizing the MI between cluster assignments has two important drawbacks. Firstly, it assumes that the number of clusters is known in advance and that these clusters are balanced (i.e., represent regions of the same size in the image). Secondly, as it involves minimizing entropy, the direct optimization of MI often leads to poor local minima, as the network becomes quickly confident in cluster assignments that are not useful (see the clusters in Figure \ref{fig:alpha_explain}, for $\alpha\!\!=\!\!0$). Our first loss term addresses this problem by combining two complementary objectives: 1) minimizing the entropy of the cluster assignment joint distribution, which encourages clusters to be balanced and confident but is also flexible to ignore some irrelevant clusters; 2) making the matrix of this joint distribution close to a diagonal matrix, which helps the optimization avoid poor minina and learn a better representation.
Another problem with the simple MI maximization approach for unsupervised representation learning is that the clusters may not align with geometric cues such as edges in the input image. In the second loss of our model, we tackle this problem by forcing the regions with high cluster entropy, which corresponds to boundaries between clusters, to be correlated with edges in the image. 
\begin{figure}[b]
 \centering
 \parbox{0.6\linewidth}{
 \includegraphics[width=1\linewidth]{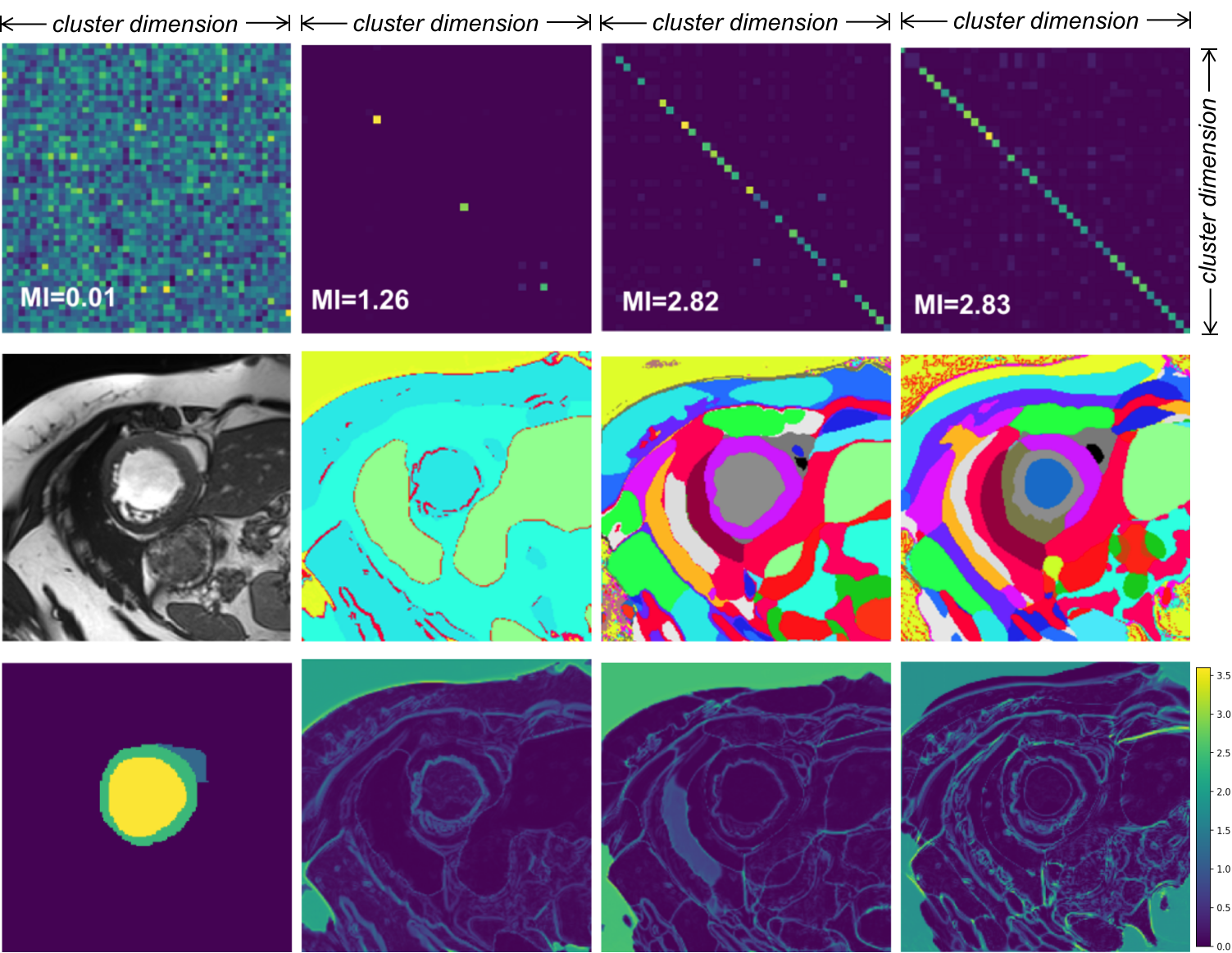} 
 \begin{tabular}{p{.21\linewidth}p{.2\linewidth}p{.19\linewidth}p{.18\linewidth}}
 & $\alpha=0.0$ & $\alpha=0.5$ & $\alpha=1.0$ 
 \end{tabular}
 }
 \caption{Influence of $\alpha$ on joint matrix $\Pjoint$ (first row), cluster assignment (second row), and the uncertainty of the cluster (third row).
 The first column shows the joint matrix before optimization, the input image and the groud-truth segmentation respectively. Using a combination of MI and cross-entropy loss ($\alpha=0.5$) provides the most meaningful unsupervised segmentation.}
 \label{fig:alpha_explain}
\end{figure}
Finally, to help the network capture the global context of images, we include a contrastive learning loss that exploits available meta-labels (e.g., slice position in a MRI volume) to make the global features obtained by the encoder $f_E$ similar for images with the same meta-label. We present a conceptual diagram of our proposed method in Fig.~\ref{fig:conceptual} of the Appendix and detail the three loss functions in the following sub-sections.

\subsection{Improved MI-based loss for \emph{dense} pre-training}

We seek to cluster the dense embeddings in feature maps $\embed$ taken from a given hidden layer of the decoder $\dec(\cdot)$. Following \citep{peng2021boosting}, we use mutual information maximization to perform clustering. The MI between two random variables $X$ and $Y$ (i.e., the cluster assignment for two images) corresponds to the KL divergence between their joint distribution $p(X,Y)$ and the product of their marginal distributions $p(X)$ and $p(Y)$:
\begin{equation}
	\MI(X,Y) \, = \, \KL\big(p(X,Y)\,||\,p(X)\,p(Y)\big)
\end{equation}
Alternatively, MI can also be defined as the difference between the combined entropy of marginals and the entropy of the joint distribution:
\begin{align}
	 \MI(X,Y) \, = &\, \entr(X) + \entr(Y) - \entr(X,Y)\nonumber\\
	  \ = &\, -\Expect_{X}[\log\Expect_Y[p(X,Y)]] -\Expect_{Y}[\log\Expect_X[p(X,Y)]]  + \Expect_{X,Y} [\log p(Y,X)]
	\label{equ:mi}
\end{align}

where $\entr(\cdot)$ is the entropy of the variable. This definition reveals that maximizing MI leads to high-entropy (uniform) distributions for $X$ and $Y$, thus avoiding trivial solutions assigning all examples to a single cluster. It also results in a low entropy of the joint distribution, corresponding to confident cluster assignments.

\begin{figure*}[t]
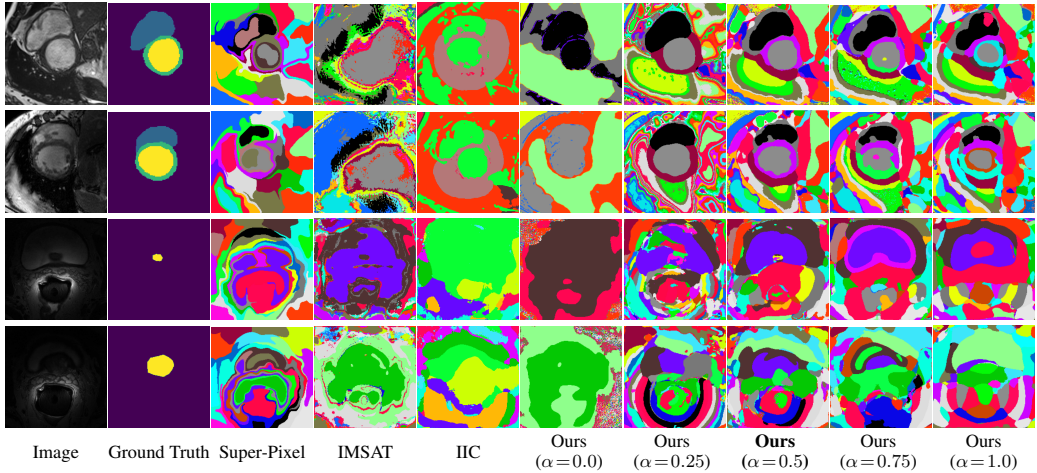

	\centering
	\tiny
	\begin{adjustbox}{max width=1\textwidth}
		\renewcommand{\arraystretch}{1}
		\setlength{\tabcolsep}{0.5pt}
		\begin{small}
			\begin{tabular}{cccccccccc}
				\visCluster{3}{acdc} \\
				\visCluster{4}{acdc} \\
				\visCluster{1}{prostate} \\
				\visCluster{2}{prostate} \\
				[-0.1mm]
				Image & Ground Truth & Super-Pixel & IMSAT & IIC  & \tabincell{c}{Ours \\ ($\alpha\!=\!0.0$)} & \tabincell{c}{Ours \\ ($\alpha\!=\!0.25$)} & \textbf{\tabincell{c}{Ours \\ ($\alpha\!=\!0.5$)}} & \tabincell{c}{Ours \\ ($\alpha\!=\!0.75$)} & \tabincell{c}{Ours \\ ($\alpha\!=\!1.0$)} 
			\end{tabular}
		\end{small}
	\end{adjustbox}
	\caption{Visual inspection of cluster assignments for unsupervised pre-training}
	\label{fig:visual_cluster_inspection}
\end{figure*}

Let $\pHat_i = \proj(s(\tsf(\img_i)))$ and \,$\pTilde_i = \proj(\tsf(s(\img_i)))$ be cluster probabilities in feature maps from a given layer of the decoder, obtained by applying a random transformation $\tsf(\cdot)$ on the input image $\img_i$ or the feature maps $s(\img_i)$. Function $\proj$ is a $1\!\times\!1$ convolutional layer followed by a $K$-way softmax projecting the feature maps to a distribution over $K$ clusters. As in IIC, we estimate the joint distribution using the average outer product between cluster probabilities $\pHat_i$ and $\pTilde_i$:
\begin{equation}
	\Pjoint =
	\frac{1}{N} \sum_{i=1}^{N} \pHat_i \cdot \pTilde_i^{\intercal}.
\end{equation}
$\Pjoint$ thus has a dimensionality of $K\!\times\!K$, and $\Pjoint^{(j,k)}$ is the joint probability of assigning $s(\tsf(\img_i))$ to cluster $j$ and $\tsf(s(\img_i))$ to cluster $k$. Following Equ. (\ref{equ:mi}), the MI between the corresponding random variables $X$ and $Y$ can be written as
\begin{equation}
\MI(X,Y) \, = \, \entropy(\mHat) + \entropy(\mTilde) - \entropy(\Pjoint) 
\label{equ:entropy}
\end{equation}
where $\mHat\!=\!\tfrac{1}{N} \sum\limits_i \pHat_i$, $\mTilde\!=\!\tfrac{1}{N} \sum\limits_i \pTilde_i$ are the cluster marginals which can computing by summing over the rows or columns of the joint distribution matrix. Maximizing the entropy of marginals encourages the network to assign an even number of samples to each cluster, and avoids trivial solutions where most clusters are empty. On the other hand, minimizing the entropy of the joint, $\entropy(\Pjoint)$, forces the network to have confident cluster assignments. 

Clustering the dense embeddings by maximizing MI poses two optimization problems. First, since we are minimizing the entropy of the joint, the network can get stuck in confident but incorrect cluster assignments which remain the same throughout optimization. Another problem stems from the fact that MI is invariant to the ordering of clusters, hence any permutation of the joint distribution matrix yields an equivalent solution. The challenge of maximizing MI is illustrated in the first row of Fig.~\ref{fig:alpha_explain}, where the left-most image is the initial joint matrix $\Pjoint$ before optimization and the one in the second column ($\alpha=0$) is $\Pjoint$ after maximizing MI. We see that only a few clusters are actually used, making the entropy of marginals low and therefore also the MI. 

To alleviate these problems, we consider the entropy of the joint distribution, given by
\begin{equation}
	\entropy(\Pjoint) \, = \, - \sum_{j=1}^{K} \sum_{k=1}^{K} \Pjoint^{(j,k)}\log \Pjoint^{(j,k)}.
\end{equation} 
A solution where cluster assignments are balanced, confident and perfectly consistent across transformations, would give a joint distribution matrix with diagonal elements equal to $1/K$ and off-diagonal elements to $0$. To guide the optimization toward this desirable solution, we introduce a pseudo-label of the joint matrix $\Ppseudo = \tfrac{1}{K}\mr{I}_K$, where $\mr{I}_K$ is the $K\!\times\!K$ identity matrix, and modify the entropy of the joint as follows:
\begin{equation}\label{eq:mod-joint-entropy}
	\entropy'_{\alpha}(\Pjoint) = - \sum_{j=1}^{K} \sum_{k=1}^{K} \!\big((1-\alpha)\cdot\Pjoint^{(j,k)} + \alpha\Ppseudo^{(j,k)}\big)\log \Pjoint^{(j,k)}
\end{equation}
In this modified formulation, $\alpha$ is a mixing coefficient ranging from 0 to 1. If $\alpha$ equals to 0, $\entropy'_\alpha(\Pjoint)$ reduces to $\entropy(\Pjoint)$, while $\alpha=1$ corresponds to a cross-entropy loss guiding the joint matrix towards the pre-defined diagonal solution $\Ppseudo$. 

Since the joint distribution matrix is computed over a batch of examples, minimizing the cross-entropy between $\Ppseudo$ and $\Pjoint$ is not the same as minimizing the cross-entropy between individual cluster assignments $\pHat_i$ and $\pTilde_i$. Nevertheless, a relationship can be derived between these two concepts, as described in the following proposition.
\begin{proposition}\label{the-proposition}
The term added in (\ref{eq:mod-joint-entropy}) corresponds the cross-entropy between the diagonal joint  $\Ppseudo\!=\!\tfrac{1}{K} \ident_K$ and $\Pjoint$, which is bounded as follows:
\begin{equation}
\log K \, \leq \, \entropy(\tfrac{1}{K}\ident_K, \Pjoint) \, \leq \, \frac{1}{N}\sum_{i=1}^N\entropy(\uu,\pHat_i) + \entropy(\uu,\pTilde_i),
\end{equation}
where $\uu$ is the vector such that $u_k = \tfrac{1}{K}$ for $k=1,\ldots,K$.
\end{proposition}
\begin{proof}
\vspace{-0.5em}
See Appendix \ref{sec:proof}.
\end{proof}

Our proposed MI loss can be thus expressed as
\begin{equation}
	\loss{MI} \, = \, -\MI'_{\alpha}(\Pjoint)\, = \,
	\entropy'_{\alpha}(\Pjoint) 
	- \entropy(\mHat) - \entropy(\mTilde)  
\end{equation}
As we will show in experiments, purely minimizing the cross-entropy between $\Ppseudo$ and $\Pjoint$ (i.e., using $\alpha=1$) does not give optimal results. This is because the true number of clusters is not known, and forcing an arbitrary number of clusters to be balanced is too restrictive. By using a value of $\alpha$ between 0 and 1, as shown in Figure \ref{fig:alpha_explain}, enables the network to ignore non-relevant clusters and focus on the most important ones. 

\subsection{Boundary-aware alignment loss for \emph{dense} feature clustering}

Clustering dense embeddings based on $\loss{MI}$ results in balanced and confident clusters, but these clusters do not need to be spatially regular or align with region boundaries in the image. 
To be useful for the downstream segmentation task, a good representation should capture anatomic structures in the images, whose contours often correspond to regions with strong intensity gradients (i.e., edges). Based on this idea, we propose to use local cross-correlation to match the boundaries of clusters, which correspond to regions with high entropy, with edges in the image. Our cross-correlation loss is defined as follows:
\begin{equation}
	\loss{CC} = \sum\limits_{i\in \Omega} 
	\frac{\bigg(\sum\limits_{j\in \mathcal{N}(i)} \big(\phi_j - \hat{\phi}(i)\big)\cdot \big(\varphi_j - \hat{\varphi}(i)\big)\bigg)\rule{0pt}{10pt}^2}
	{
		\bigg(\sum\limits_{j\in \mathcal{N}(i)}\!\! \big(\phi_j - \hat{\phi}(i)\big)^2\bigg) \!\cdot\! \bigg(\sum\limits_{j\in \mathcal{N}(i)}\!\! \big(\varphi_j - \hat{\varphi}(i)\big)^2\bigg)
	}
\end{equation}
In this loss, $\phi$ measures the edge response of a Sobel filter on the input image, while $\varphi$ is a spatial map of cluster distribution entropy. $\hat{\phi}(i)$ and  $\hat{\varphi}(i)$ denote the mean value in a local window $\mathcal{N}(i)$ centered on position $i$, respectively for $\phi$ and $\varphi$.  
We note that a similar loss is often used in medical image registration~\citep{balakrishnan2019voxelmorph}, where images of two different modalities or acquisitions need to be aligned. Unlike $L_2$ loss, which imposes a strict equivalence between distributions, this loss can capture correlation in local variance even when images have very different distributions of intensity.
 
\subsection{Contrastive loss for \emph{global} feature learning}

While the first two losses aim to regularize the local representation of dense feature maps in the decoder, the next one focuses on learning a global representation of the image. Toward this goal, we consider the features produced by the encoder $\enc(\img_i)$, which summarize the global context of an input image $\img_i$, and project it into a low-dimensional representation $\zz_i$. 
Similar to \citep{chaitanya2020contrastive}, we regularize global representation $\zz_i$ using a contrastive loss exploiting available meta-labels:
\vspace{-6pt}
\begin{equation}\label{equ:sup-con}
	\loss{con} \, = \, - \frac{1}{2N}\sum_{i=1}^{2N}\frac{1}{|\mathcal{S}_i|}\sum_{j\in \mathcal{S}_i} \log \frac{\exp\big(\tr{\zz}_i \zz_j/\tau\big)}{
		\sum\limits_{a \in \mathcal{S} \setminus \{i\}} \!\!\!\exp\big(\tr{\zz}_i \zz_a/\tau\big)}. 
\end{equation}
In the loss, $\mathcal{S} = \{i\, | 1\leq i \leq 2N\}$ is the index set of an augmented batch, where each image is randomly transformed twice. Moreover, $\mathcal{G}(i)$ the meta-label of image $i$ and $\mathcal{S}_i = \{j \, | \, \mathcal{G}(j) = \mathcal{G}(i), \, 1\leq j \leq 2N, \, i\neq j \}$ are the indexes of images within the same meta-label as $i$. As described in Section \ref{sec:comparison_methods}, we divide volumetric images into different partitions, and use the partition index of each 2D image (slice in the volume) as meta-label. $\tau$ is a small temperature factor that helps gradient descent optimization by smoothing the landscape of the loss objective.
\subsection{Our unified pre-training objective}
Our final objective for unsupervised representation learning combines all three objectives as follows:
\begin{equation}
	\loss{total} = \!\! \underbrace{\loss{con}}_{\text{global embedding}} \!\! + \, \underbrace{ \lambda \loss{CC} \, + \, \loss{MI} }_{\text{dense features}} 
	\label{equ:total}
\end{equation}
$\loss{con}$ is applied on the global representation of the encoder, and therefore it influences only the encoder. Conversely, $\loss{MI}$ and $\loss{CC}$ are used on the dense features of the penultimate layer of the decoder, hence they affect the parameters of the whole network. These three losses are learned jointly in a single pre-training step. 
See Fig.~\ref{fig:conceptual} in the Appendix for a graphical illustration.

\begin{table*}[t!]
	\centering
	\caption{3D DSC on test set when fine-tuned using a few labeled data. Listed methods are applied in a pre-training stage. (\emph{Dec}) means that the loss is applied to dense embeddings in feature maps of the decoder, and (\emph{Enc}) to the global features at the end of the encoder.
	}
\label{tab:representation_result}
	\begin{adjustbox}{max width=1\textwidth}
		\setlength{\tabcolsep}{2pt}
		\renewcommand{\arraystretch}{1.0}
		\begin{tabular}{p{0.23\textwidth}|cccc|cccc|cccc|cccc}
			\toprule
			\multirow[b]{2}{*}{\textbf{Methods}} 
			& \multicolumn{4}{c|}{\textbf{ACDC-LV}}
			& \multicolumn{4}{c|}{\textbf{ACDC-RV}}
			& \multicolumn{4}{c|}{\textbf{ACDC-Myo}}
			& \multicolumn{4}{c}{\textbf{P\textsc{romise}12}}\\
			\cmidrule(l{0pt}r{0pt}){2-5} 
			\cmidrule(l{0pt}r{0pt}){6-9} 
			\cmidrule(l{0pt}r{0pt}){10-13} 
			\cmidrule(l{0pt}r{0pt}){14-17} 
			 & \bfseries 1 scan & \bfseries 2 scans & \bfseries 4 scans & \bfseries mean & \bfseries 1 scan & \bfseries 2 scans & \bfseries 4 scans & \bfseries mean & \bfseries 1 scans & \bfseries 2 scans & \bfseries 4 scans & \bfseries mean & \bfseries 4 scans & \bfseries 6 scans & \bfseries 8 scans & \bfseries mean \\
			\midrule
			Baseline 
			& \res{67.13} {(2.27)}& \res{74.49} {(2.24}& \res{84.81} {(2.47)} & 75.48
			& \res{51.82} {(5.17)}& \res{60.50} {(3.68)}& \res{64.18} {(2.33)} & 58.84 
			& \res{54.05} {(0.62)}& \res{67.56} {(1.07)}& \res{76.00} {(0.13)} & 65.87
			& \res{49.91} {(1.69)} & \res{71.53} {(2.30)} & \res{78.04} {(1.82)} & 66.49
			\\
			Full Sup. &\multicolumn{4}{c|}{\res{92.26} {(0.37)}}& \multicolumn{4}{c|}{\res{86.80} {(0.09)}} & \multicolumn{4}{c|}{\res{88.07} {(0.28)}} 
			& \multicolumn{4}{c}{\res{89.65} {(0.96)}}\\
			\midrule 
			IIC (\emph{Dec})
			& \res{71.96} {(1.63)}& \res{82.84} {(0.99)}& \res{85.43} {(1.09)} & 80.08 
			& \res{56.92} {(4.00)}& \res{63.58} {(2.41)}& \res{64.93} {(3.50)} & 61.81
			& \res{58.31} {(3.73)}& \res{70.22} {(3.49)}& \res{74.98} {(1.74)} & 67.83 
			& \res{54.21} {(3.79)} & \res{72.97} {(1.41)} & \res{80.05} {(0.01)} & 69.07 \\
			IMSAT (\emph{Dec})
			& \res{57.59} {(9.18)}& \res{75.38} {(1.13)} & \res{76.76} {(6.18)}& 69.91 
			& \res{34.36} {(8.74)}& \res{47.81} {(3.41)} & \res{47.42} {(3.40)}& 43.20
			& \res{50.52} {(1.51)}& \res{64.51} {(3.36)} & \res{71.91} {(2.35)}& 62.23
			& \res{55.20} {(7.84)} & \res{74.46} {(2.48)} & \res{81.50} {(1.11)} & 70.38 
			\\
			
			\tabincell{l}{Contrast (\emph{Dec}) } 
			& \res{64.37} {(3.18)}& \res{77.69} {(0.98)}& \res{84.36} {(2.60)} & 75.47
			& \res{50.75} {(2.91)}& \res{56.34} {(2.57)}& \res{50.88} {(2.79)} & 52.66
			& \res{54.40} {(1.92)}& \res{70.11} {(0.88)}&\res{74.05} {(1.65)} & 69.19 
			& \res{54.22} {(2.44)} & \res{63.52}{(7.10)} & \res{82.47} {(1.03)} & 66.74 \\

			Ours (\emph{Dec}) (only MI) 
			& \res{83.63} {(0.89)}& \res{86.94} {(0.97)}& \underline{\res{89.33} {(0.52)}} & 86.63
			& \res{66.63} {(1.06)}& \res{73.78} {(1.90)}& \res{73.85} {(1.36)} & 71.42
			& \res{73.65} {(2.33)}& \res{77.39} {(2.47)}& \res{81.92} {(1.16)} & 77.08 
			& \res{65.36} {(2.40)} & \res{78.42} {(3.14)} & \res{81.92} {(2.89)} & 75.23 \\ 
			Ours (\emph{Dec}) (only CC) 
			& \res{64.85} {(6.14)}& \res{67.03} {(6.59)}& \res{79.31} {(7.19)} & 70.70
			& \res{44.30} {(4.65)}& \res{50.33} {(5.27)}& \res{54.52} {(3.46)} & 49.72
			& \res{49.46} {(1.04)}& \res{60.13} {(1.51)}& \res{69.64} {(4.20)} & 59.74
			& \res{42.48} {(8.72)} & \res{73.69} {(3.22)} & \res{80.31} {(2.04)} & 65.50 \\
			
			Ours (\emph{Dec}) (MI+CC) 
			& \underline{\res{84.04} {(1.07)}} & \textbf{\res{88.52} {(1.35)}} & \res{89.31} {(0.58)}& \underline{87.29}& 
			\underline{\textbf{\res{76.86} {(0.49)}}}& \underline{\res{79.13} {(0.94)}} & \res{75.92 }{(1.84)} & \underline{77.30} 
			& \textbf{\res{76.93} {(1.58)}} & \textbf{\res{79.59} {(0.85)}} & \underline{\res{81.97} {(0.59)}} & \textbf{79.49} 
			& \underline{\res{68.13}{(2.05)}} & \underline{\res{78.75} {(1.61)}} & \textbf{\res{82.82} {(3.25)}} & \underline{76.30}
			\\
			\midrule
						
			\tabincell{l}{Contrast (\emph{Enc}) } 
			& \res{80.59} {(0.86)}& \res{85.68} {(0.59)}& \res{87.78} {(0.20)} &84.10
			& \res{68.91} {(3.79)}& \res{73.54} {(0.18)}& \res{72.70} {(2.57)} & 71.72
			& \res{67.30} {(0.97)}& \res{77.22} {(1.31)} & \res{79.58} {(2.38)}& 74.70 
			& \res{63.54} {(2.08)} & \res{78.24} {(2.02)} & \res{81.72} {(0.90)} & 74.50\\
			\tabincell{l}{Contrast (\emph{Enc}+\emph{Dec}) } 
			& \res{77.98} {(2.90)} & \res{85.97} {(2.11)} & \res{88.42} {(1.20)} & 84.12
			& \res{66.47} {(1.21)}& \res{72.82} {(1.32)} & \underline{\res{76.69} {(3.23)}} & 71.99 
			& \res{64.96} {(0.61)} & \res{76.98} {(1.93)} & \res{78.76} {(0.68)} & 73.57 
			& \res{60.68} {(7.47)} & \res{77.97} {(4.70)} & \res{80.53} {(2.23)} & 73.06 \\
			\tabincell{l}{Contrast (\emph{Enc})+Ours (\emph{Dec})} 
			& \textbf{\res{84.48} {(1.32)}} & \underline{\res{87.85} {(0.45)}} & \textbf{\res{90.04 }{(1.31)}} & \textbf{87.45} 
			& \underline{\res{75.42} {(3.10)}} & \textbf{\res{79.73} {(1.66)}} & \textbf{\res{78.89} {(1.95)}} & \textbf{78.01}
			& \underline{\res{74.30} {(1.76)}} & \underline{\res{78.43} {(1.52)}} & \textbf{\res{82.82} {(0.97)}} & \underline{78.52} 
			& \textbf{\res{69.76} {(2.46)}} & \textbf{\res{80.47} {(0.76)}} & \underline{\res{82.09} {(2.47)}} & \textbf{77.44}
			\\
			\bottomrule
		\end{tabular}
	\end{adjustbox}
\end{table*}

\section{Experimental setup}
To assess the performance of our proposed pre-training method, we performed extensive experiments on two clinically-relevant segmentation datasets. In this section, we present briefly the experimental setting, employed dataset, as well as implementation details. We include more details in Appendix \ref{sec:append_dataset} and \ref{sec:implementation_details_append}. 
\subsection{Dataset and evaluation metrics}

Two clinically-relevant benchmark dataset are chosen for our experiments: the automatic cardiac diagnosis challenge (ACDC)~\citep{ACDC}, and the Prostate MR image segmentation 2012 challenge (P\textsc{romise}12)~\citep{litjens2014evaluation} dataset. Three foreground classes are delineated for ACDC dataset, which includes left ventricle endocardium (LV), left ventricle myocardium (Myo), right ventricle endocardium (RV), and we consider them as three binary segmentation tasks.
Due to the high anisotropic resolution in both datasets, we consider the 2D slices of volumetric images as separate examples, and randomly split them into training, validation and test sets, so that no two images of the same scan are in the same set.  
To evaluate methods in a setting with limited annotation, we randomly select images from a few scans of the training set as our labeled data set, and consider all the images of the training set as unlabeled. 
We detail the data pre-processing and augmentation in Appendix \ref{sec:append_dataset}. 
For all datasets, we used the 3D Dice similarity coefficient (DSC), which  
measures the overlap between the predicted labels $S$ and the corresponding ground truth labels $G$:
$\text{DSC}(S,G) = \frac{2\times|S\cap G |}{|S|+|G|}$.
In all experiments, we reconstruct the 3D segmentation for each scan by aggregating the predictions made for 2D slices and report the 3D DSC metric for the test set corresponding to the best-performing epoch on the validation set. 

\subsection{Comparable methods and ablation variants}
We compare our proposed method with clustering-based and contrastive-based self-supervised learning approaches, as well as six state-of-the-art semi-supervised segmentation methods. IMSAT~\citep{hu2017learning} and IIC~\citep{ji2019invariant} also employ MI maximization as the optimization criterion and cluster the local embeddings pixel-wisely. The contrastive learning method relies on the construction of positive and negative pairs. 
For dense embedding, positive pairs are embeddings in the same position undergoing different intensity transformations, where negative pairs are defined as embeddings with sufficient large distances. Six semi-supervised segmentation methods are also tested, and we present the details of each method in Appendix~\ref{sec:comparison_methods}.
Lastly, we boost the best performing semi-supervised method with pre-trained weights from our pre-training methods. 

\subsection{Implementation details}

We employ U-Net~\citep{ronneberger2015unet} as our segmentation network architecture, which consists of five symmetric encoder/decoder blocks with skip connections. We extract the local embeddings in the decoder layer before the last $1\!\times\!1$ convolution, thus they have the same spatial resolution as the input image. These embeddings are then projected to probabilities over $K$ cluster using a projector comprised of a $1\!\times\!1$ convolution and a $K$-way softmax. We fix $K\!=\!40$ for all datasets. Hyper-parameter
$\alpha$ is introduced in our method and we fixed it to 0.5 for all experiments. Image transformation $\mathcal{T}(\cdot)$ consists of gamma correction and random affine transformation. Our proposed method follows a two-stage training strategy: \emph{pre-train} for representation learning and \emph{fine-tune} for evaluation this representation on the downstream segmentation task. In the \emph{pre-train} stage, we optimize the network in an unsupervised way on all training images without pixel-wise annotation, resulting in a set of network parameters $\theta$. We evaluate the quality of these pre-trained weights in a separate \emph{fine-tune} stage by creating a second segmentation network initialized with these parameters, and fine-tuning the whole network using only a few labeled scans. 
The comparison with other SOTA semi-supervised methods is performed with the same setting as in the \emph{fine-tune} stage, and we report their test DSC performances on their own best hyper-parameters determined by validation performance using grid search.  We provide detailed explanation on network architecture, training protocols, transformation $\mathcal{T}(\cdot)$, and hyper-parameters used in each method in Appendix \ref{sec:implementation_details_append}.

\section {Experimental results}\label{sec:results}

In this section, we first compare our method against clustering-based and contrastive-based methods.
Then, we evaluate all components of our method as our ablation variants. Finally, we compare our method with the most promising approaches for semantic segmentation in medical imaging, with reduced training data.

\subsection{Comparison with cluster based methods and ablation variants}
\label{sec:results1}
Table \ref{tab:representation_result} reports the test 3D DSC performance for different representation learning methods on the ACDC and P\textsc{romise}12 datasets. 
At the top of the table, we report the number of labeled scans used for every result. Reported values are the average over \emph{three} independent runs with different random seeds. Methods presented here all adopt a \emph{pre-train} and \emph{fine-tune} strategy with a few annotated scans. 

\emph{Upper and lower bounds:} We present results for \emph{Baseline}, which uses only the annotated scans with cross-entropy as standard supervised loss, and for \emph{Full Supervision}, where the same loss is used with all available training examples. These represent lower and upper bounds on the expected performance for different methods.

\emph{Cluster-based methods:} We present in the next two rows the performance for IIC and IMSAT. These two methods employ MI as the optimization objective
and perform clustering on local embeddings with  $K$ clusters.
IIC brings consistent improvements across all four tested classes (4.6\%, 3.14\%, 1.96\%, and 2.58\%), while IMSAT leads to a worse performance for the ACDC dataset. We visualize their pre-trained clusters in Fig. \ref{fig:visual_cluster_inspection}, showing that these methods fail to find balanced clusters corresponding to meaningful regions of the image.

\emph{Contrastive-based method:} We then report in the next row the performance obtained using contrastive learning only on dense features of the decoder. Surprisingly, we observe that optimizing the contrastive objective with grid-based positive and negative pairs provides no benefit for the segmentation tasks. This is due to very weak guidance offered by contrasting dense embeddings. 

\emph{Our ablations:} We then present in the next three lines the performance for our proposed ablation variants. Our modified $\loss{MI}$ alone leads to substantial improvements compared to the original IIC: 6.54\%, 9.61\%, 9.25\%, and 6.16\% are observed for the four classes. These improvements clearly indicate the advantage of introducing a pseudo-mask $\Ppseudo$ to guide the learning of the joint probability matrix. $\loss{CC}$ aligns cluster boundaries with image edges, but does not help segmentation on its own since predicted clusters are not consistent across images and transformations. Last, we observe that combining our proposed $\loss{MI}$ and $\loss{CC}$ lead to significant improvements over using $\loss{MI}$ alone. These improvements are particularly notable for RV and Myo classes, which are more complex and rely more on image edges. 

\emph{Global feature pre-training:} The last three rows report the performance of methods employing contrastive learning on global features. The method Contrast (Enc) which only optimizes $\loss{con}$ significantly improves the segmentation quality given a few labeled scans. However, these improvements are still inferior to the Ours (Dec) (MI+CC) variant which, unlike Contrast (Enc), does not use meta-labels. Contrast (Enc+Dec), which combines \emph{global} and \emph{local} contrastive objectives, leads to marginal improvements. Our proposed method is complementary to the \emph{global} contrastive based method. We report in the last row the performance of our proposed method combining all three losses: $\loss{con}$, $\loss{MI}$ and $\loss{CC}$. This method achieves the highest accuracy on 10 out of 16 cases, and second rank for remaining cases. Further, it yields average DSC improvement over Baseline as large as 17.35\%, 23.60\%, 20.25\% and 17.85\%, for the LV, RV, Myo and Prostate tasks, respectively. 
\begin{table}[t]
	\centering
	\caption{Impact of $\alpha$ for our proposed $\loss{MI}$.}
	\label{tab:alpha_ablation}
	\begin{adjustbox}{max width=0.65\linewidth}
		\setlength{\tabcolsep}{3pt}
		\renewcommand{\arraystretch}{1}
		\begin{tabular}{c|ccc|ccc|ccc}
			\toprule
			\multirow[b]{2}{*}{\phz$\boldsymbol{\alpha}$\phz~} 
			& \multicolumn{3}{c|}{\textbf{ACDC-LV}}
			& \multicolumn{3}{c|}{\textbf{ACDC-RV}}
			& \multicolumn{3}{c}{\textbf{ACDC-Myo}}\\
			\cmidrule(l{0pt}r{0pt}){2-4} 
			\cmidrule(l{0pt}r{0pt}){5-7} 
			\cmidrule(l{0pt}r{0pt}){8-10} 
			 & \bfseries 1 scan & \bfseries 2 scans & \bfseries 4 scans & \bfseries 1 scan & \bfseries 2 scans & \bfseries 4 scans & \bfseries 1 scans & \bfseries 2 scans & \bfseries 4 scans \\
			\midrule
					
			$0.0\phz$ & \res{61.58 }{(2.36)} & \res{78.09 }{(4.95)} & \res{81.27 }{(5.96)} & \res{31.95 }{(5.23)} & \res{23.50 }{(3.25)} & \res{41.64 }{(7.96)} & \res{54.74 }{(4.72)} & \res{58.15 }{(8.67)} & \res{69.45 }{(3.38)} \\

			$0.25$ & \textbf{\res{84.36 }{(0.90)}} & \res{87.68 }{(0.59)} & \textbf{\res{89.32 }{(0.14)} } & \res{38.98 }{(23.06)} & \res{59.73 }{(7.74)} & \res{59.39 }{(4.83)} & \textbf{\res{76.93 }{(1.58)}} & \res{79.59 }{(0.85)} & \res{81.97 }{(0.59)} \\
			$0.5\phz$ & \res{84.04 }{(0.87)} & \textbf{\res{88.52 }{(1.11)}} & \res{89.31 }{(0.48)} & \textbf{\res{76.86 }{(0.35)}} & \res{79.13 }{(0.66)} & \textbf{\res{75.92 }{(1.30)}} & \res{76.27 }{(1.32)} & \textbf{\res{79.81 }{(0.86)}} & \textbf{\res{82.36 }{(0.92)}} \\
			$0.75$ & \res{82.02 }{(0.81)} & \res{87.81 }{(0.51)} & \res{89.03 }{(0.54)} & \res{76.76 }{(0.57)} & \textbf{\res{79.41 }{(0.59)}} & \res{75.49 }{(4.02)} & \res{73.14 }{(0.74)} & \res{78.79 }{(1.39)} & \res{81.79 }{(0.23)} \\
			$1.0\phz$ & \res{81.31 }{(0.70)} & \res{85.58 }{(0.70)} & \res{88.66 }{(0.51)} & \res{73.34 }{(1.65)} & \res{76.37 }{(1.14)} & \res{70.44 }{(1.56)} & \res{71.90 }{(3.72)} & \res{79.47 }{(1.09)} & \res{81.22 }{(0.90)} \\
			\bottomrule
		\end{tabular}
	\end{adjustbox}	 
\end{table}

\begin{table}[t]
	\centering
	\caption{Impact of our proposed boundary-aware loss $\loss{CC}$.}
	\label{tab:impact_boundary}
	\begin{adjustbox}{max width=0.65\linewidth}
		\setlength{\tabcolsep}{3pt}
		\renewcommand{\arraystretch}{1}
		\begin{tabular}{c|ccc|ccc|ccc}
			\toprule
			\multirow[b]{2}{*}{\phz$\boldsymbol{\lambda_{\mr{CC}}}$\phz} 
			& \multicolumn{3}{c|}{\textbf{ACDC-LV}}
			& \multicolumn{3}{c|}{\textbf{ACDC-RV}}
			& \multicolumn{3}{c}{\textbf{ACDC-Myo}}\\
			\cmidrule(l{0pt}r{0pt}){2-4} 
			\cmidrule(l{0pt}r{0pt}){5-7} 
			\cmidrule(l{0pt}r{0pt}){8-10} 
			 & \bfseries 1 scan & \bfseries 2 scans & \bfseries 4 scans & \bfseries 1 scan & \bfseries 2 scans & \bfseries 4 scans & \bfseries 1 scans & \bfseries 2 scans & \bfseries 4 scans \\
			\midrule
					
 $0.0$ & 83.63 & 86.94 & 89.33 & 66.63 &73.78 &73.85 &73.65 &77.39 &81.92 \\

			$0.1$ & 80.54 & 85.40 & 87.80 & 66.99 & 76.00 & 75.33 & 72.62 & 77.03 & 81.39 \\
			$1.0$ & \textbf{84.04} & \textbf{88.52} & \textbf{89.31} & \textbf{76.86} & \textbf{79.13} & \textbf{75.92} & \textbf{76.93} & \textbf{79.59}& \textbf{81.97} \\
			$4.0$ & 78.36 & 84.60 & 85.48 & 70.84 & 75.39 & 70.88 & 73.74 & 76.65 & 81.33 \\
			
			\bottomrule
		\end{tabular}
	\end{adjustbox}
\end{table}
\begin{figure}[b]
	\centering
	\vspace{0.8em}
	\begin{adjustbox}{max width=0.75\linewidth}
		\renewcommand{\arraystretch}{1}
		\setlength{\tabcolsep}{1pt}
		\begin{small}
			\begin{tabular}{ccccc}
				\includegraphics[width=0.2\linewidth]{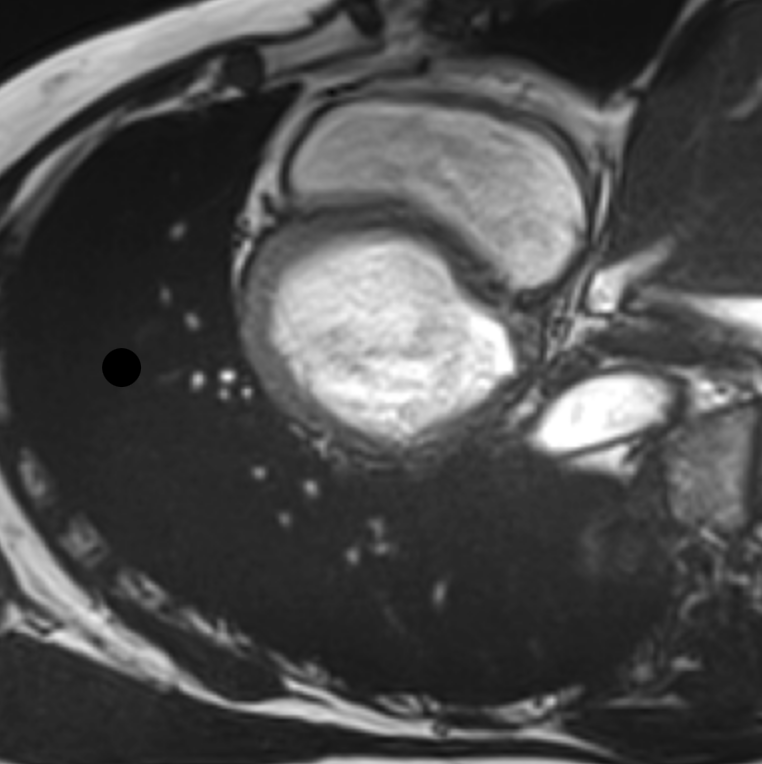} & 
				\includegraphics[width=0.2\linewidth]{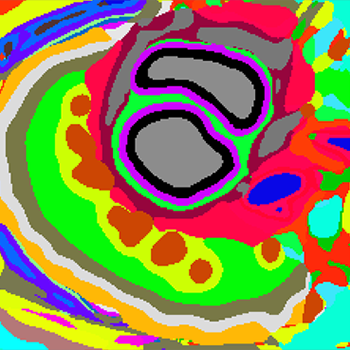} & 
				\includegraphics[width=0.2\linewidth]{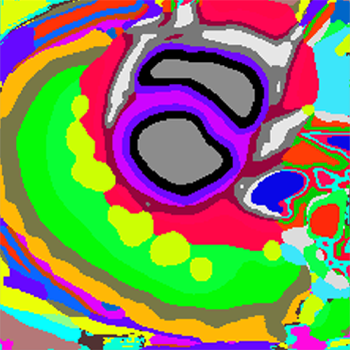} & 
				\includegraphics[width=0.2\linewidth]{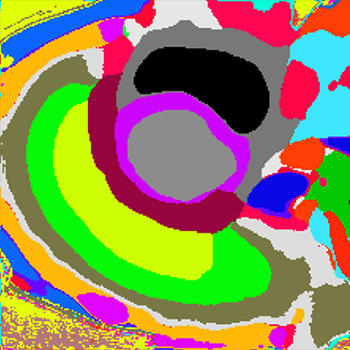} & 
				\includegraphics[width=0.2\linewidth]{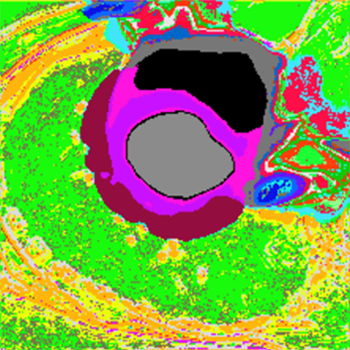} \\
							
				\includegraphics[width=0.2\linewidth]{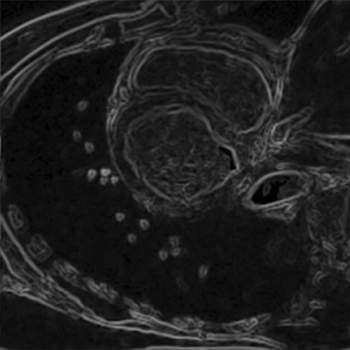} & 
				\includegraphics[width=0.2\linewidth]{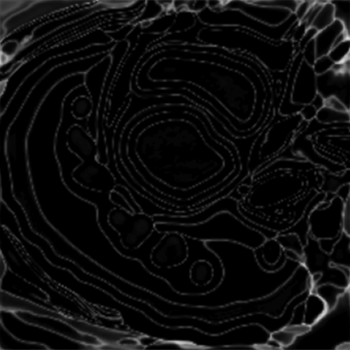} & 
				\includegraphics[width=0.2\linewidth]{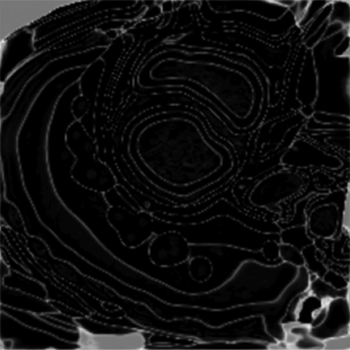} & 
				\includegraphics[width=0.2\linewidth]{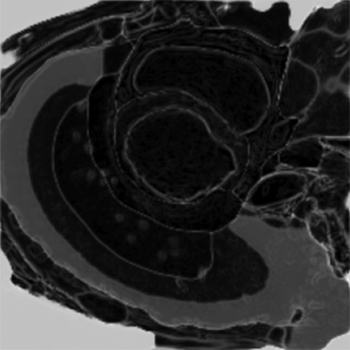} & 
				\includegraphics[width=0.2\linewidth]{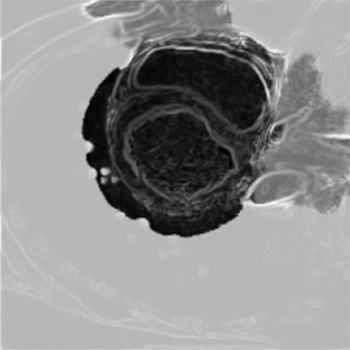} \\
				\tabincell{c}{Image / \\ 
				Gradient}  & \tabincell{c}{
				$\lambda_{\text{CC}}\!=\!0.0$} 
				  & \tabincell{c}{
				$\lambda_{\text{CC}}\!=\!0.1$} 
				  & \tabincell{c}{
				$\lambda_{\text{CC}}\!=\!1.0$} 
				  & \tabincell{c}{
				$\lambda_{\text{CC}}\!=\!4.0$} 
			\end{tabular}
		\end{small}
	\end{adjustbox}
	\caption{Boundary loss effect. Upper: Input image and different pre-trained clusters; Down: Image edges and entropy map for each cluster. Black color refers to certain cluster regions while bright regions reflect uncertain predictions.}
	\label{fig:boundary}
\end{figure}

\begin{table}[t]
	\centering
	\caption{Impact of number of clusters $K$.}
	\label{tab:impact_k}
	\begin{adjustbox}{max width=0.65\linewidth}
		\setlength{\tabcolsep}{3pt}
		\renewcommand{\arraystretch}{1}
		\begin{tabular}{c|ccc|ccc|ccc}
			\toprule
			\multirow[b]{2}{*}{\phz$\boldsymbol{K}$\phz} 
			& \multicolumn{3}{c|}{\textbf{ACDC-LV}}
			& \multicolumn{3}{c|}{\textbf{ACDC-RV}}
			& \multicolumn{3}{c}{\textbf{ACDC-Myo}}\\
			\cmidrule(l{0pt}r{0pt}){2-4} 
			\cmidrule(l{0pt}r{0pt}){5-7} 
			\cmidrule(l{0pt}r{0pt}){8-10} 
			 & \bfseries 1 scan & \bfseries 2 scans & \bfseries 4 scans & \bfseries 1 scan & \bfseries 2 scans & \bfseries 4 scans & \bfseries 1 scans & \bfseries 2 scans & \bfseries 4 scans \\
			\midrule
			\phz5 & 66.88 & 73.95 & 77.33 
			& 49.32 & 52.10 & 51.74 
			& 38.52 & 48.18 & 49.25 \\
			10 & 81.00 & 86.11 & 87.73 
			& 70.75 & 74.37 & 71.78 
			& 71.16 & 77.48 & 80.14 \\
			20 & 69.72 & 73.37 & 74.68 
			& 59.63 & 63.13 & 54.00 
			& 76.03 & 79.34 & 82.30 \\
			40 & 84.04 & \textbf{88.52} & 89.31 
			& \textbf{76.86} & \textbf{79.13} & 75.92 
			& \textbf{76.27} & \textbf{79.81} & 82.36 \\
			60 & \textbf{85.25} & 88.12 & \textbf{89.74} 
			& 70.28 & 74.78 & \textbf{76.39} 
			& 75.96 & 79.09 & \textbf{82.68} \\
			\bottomrule
		\end{tabular}
	\end{adjustbox}
	 
\end{table}

\subsection{Visualization of pre-trained cluster assignments}

To better understand our boundary-aware information maximization method, we visualize in Fig.~\ref{fig:visual_cluster_inspection} different cluster assignments obtained by our ablation variants and compared methods. Clusters obtained at the end of the unsupervised pre-training are illustrated by different colors. We also compare these clusters with the SLIC super-pixel algorithm \citep{achanta2012slic} that groups pixels based on both intensity and spatial information.
We note that IMSAT, IIC and Ours ($\alpha=0$) produce highly unbalanced clusters, where a few clusters dominate a large portion of pixels and resulting clusters do not correspond well to anatomical structures of the image. In contrast, our proposed variants with $\alpha\geq 0.25$ clearly capture the main structures in cardiac MR images, without any pixel-wise annotation. In most cases, it is able to successfully separate the LV, RV and Myo classes from the background. Additionally, P\textsc{romise}12 images present less contrast but our method still produces relatively better anatomical structures compared with traditional Super-pixel methods, IIC and IMSAT. 
This explains the huge improvements brought by our method when fine-tune the network using a few labeled images. Last, we notice that contrastive-based approaches can also boost the segmentation performance. However, they lack the intepretability of clusters provided by our method.
\begin{table*}[t!]
	\centering
	\caption{We compare the performance of our method with other pre-training approaches and state-of-the-art semi-supervised methods on 3D DSC on test set when fine-tuned using a few labeled data. }
	\label{tab:compare_sota}
	\begin{adjustbox}{max width=1\textwidth}
		\setlength{\tabcolsep}{2pt}
		\renewcommand{\arraystretch}{1.0}
		\begin{tabular}{p{0.23\textwidth}|cccc|cccc|cccc|cccc}
			\toprule
			\multirow[b]{2}{*}{\textbf{Methods}} 
			& \multicolumn{4}{c|}{\textbf{ACDC-LV}}
			& \multicolumn{4}{c|}{\textbf{ACDC-RV}}
			& \multicolumn{4}{c|}{\textbf{ACDC-Myo}}
			& \multicolumn{4}{c}{\textbf{P\textsc{romise}12}}\\
			\cmidrule(l{0pt}r{0pt}){2-5} 
			\cmidrule(l{0pt}r{0pt}){6-9} 
			\cmidrule(l{0pt}r{0pt}){10-13} 
			\cmidrule(l{0pt}r{0pt}){14-17} 
			& \bfseries 1 scan
			& \bfseries 2 scans
			& \bfseries 4 scans
			& \bfseries mean
					
			& \bfseries 1 scan
			& \bfseries 2 scans
			& \bfseries 4 scans
			& \bfseries mean
					
			& \bfseries 1 scans
			& \bfseries 2 scans
			& \bfseries 4 scans
			& \bfseries mean
					 
			& \bfseries 4 scans
			& \bfseries 6 scans
			& \bfseries 8 scans
			& \bfseries mean
			\\
			\midrule
			Baseline 
			 & \res{67.13} {(2.27)} & \res{74.49} {(2.24)} & \res{84.81} {(2.47)} & 75.48 
			 & \res{51.82} {(5.17)} & \res{60.50} {(3.68)} & \res{64.18} {(2.33)} & 58.84 
			 & \res{54.05} {(0.62)} & \res{67.56} {(1.07)} & \res{76.00} {(0.13)} & 65.87 
			 & \res{49.91} {(1.69)} & \res{71.53} {(2.30)} & \res{78.04} {(1.82)} & 66.49 \\
			\midrule
			\tabincell{l}{Contrast (\emph{Enc}+\emph{Dec}) } 
			& \res{77.98} {(2.90)} & \res{85.97} {(2.11)} & \res{88.42} {(1.20)} & 84.12
			& \res{66.47} {(1.21)}& \res{72.82} {(1.32)} & \res{76.69} {(3.23)} & 71.99 
			& \res{64.96} {(0.61)} & \res{76.98} {(1.93)} & \res{78.76} {(0.68)} & 73.57 
			& \res{60.68} {(7.47)} & \res{77.97} {(4.70)} & \res{80.53} {(2.23)} & 73.06 \\
						Ours (\emph{pre-train}) 			& \res{84.48} {(1.32)} & \res{87.85} {(0.45)} & \res{90.04 }{(1.31)}& 87.45 
			& \res{75.42} {(3.10)} & \res{79.73} {(1.66)} & \res{78.89} {(1.95)} & 78.01
			& \res{74.30} {(1.76)} & \res{78.43} {(1.52)} & \res{82.82} {(0.97)} & 78.52 
			& \res{69.76} {(2.46)} & \res{80.47} {(0.76)} & \res{82.09} {(2.47)} & 77.44\\
			\midrule
			\tabincell{l}{Entropy Min.}
			 & \res{73.79}{(2.79)} & \res{80.26}{(0.80)} & \res{86.84}{(0.87)}& 80.30 			 & \res{56.18}{(2.59)} & \res{62.09}{(3.32)} & \res{66.27}{(1.40)} & 61.51 
			 & \res{57.23}{(1.90)} & \res{71.10}{(1.02)} & \res{76.28}{(1.20)} & 68.20
			 & \res{59.78}{(0.76)} & \res{76.09}{(0.52)} & \res{78.98}{(2.58)} & 71.62 
			\\
			\tabincell{l}{MixUp} 
			 & \res{73.30}{(2.31)} & \res{76.30}{(3.14)} & \res{84.42}{(0.29)} & 78.01
			 & \res{61.23}{(1.88)} & \res{63.60}{(1.09)} & \res{63.14}{(0.59)} & 62.66
			 & \res{55.74}{(2.90)} & \res{69.80}{(1.70)} & \res{73.84}{(1.20)} &66.46
			 & \res{52.09}{(7.62)} & \res{75.59}{(1.08)} & \res{81.11}{(2.59)} & 69.60
			\\
			\tabincell{l}{Mean Teacher (MT)}
			 & \res{83.13}{(4.52)} & \res{87.02}{(1.27)} & \res{87.70}{(0.74)} & 85.95
			 & \res{61.61}{(4.20)} & \res{68.76}{(1.57)} & \res{67.21}{(2.44)} & 65.86
			 & \res{61.55}{(2.16)} & \res{75.32}{(5.37)} & \res{78.42}{(1.13)} & 71.76
			 & \underline{\res{84.71}{(0.51)}} & \underline{\res{85.97}{(0.45)}} & \underline{\res{86.93}{(0.22)}} & \underline{85.87}
			\\
			\tabincell{l}{UA-MT} 
			 & \res{81.08}{(1.28)} & \res{85.03}{(0.35)} & \res{87.19}{(1.91)} & 84.43
			 & \res{62.06}{(1.25)} & \res{67.91}{(3.87)} & \res{66.64}{(1.39)} & 65.54
			 & \res{59.26}{(0.85)} & \res{73.68}{(1.54)} & \res{78.61}{(1.31)} & 70.52
			 & \res{66.16}{(7.73)} & \res{81.79}{(0.93)} & \res{84.40}{(1.24)} & 77.45
			\\

			\tabincell{l}{ICT} 
			 & \res{76.87}{(1.18)} & \res{78.41}{(3.68)} & \res{86.34}{(0.38)} & 80.54
			 & \res{60.31}{(2.57)} & \res{63.42}{(1.86)} & \res{68.35}{(0.70)} & 64.03
			 & \res{55.91}{(1.02)} & \res{71.77}{(0.62)} & \res{77.90}{(2.07)} & 68.53
			 & \res{63.97}{(2.07)} & \res{77.92}{(2.04)} & \res{81.39}{(0.15)} & 74.43 \\
 		\tabincell{l}{Adv. Train.} 
			 & \res{75.31}{(2.29)} & \res{74.85}{(2.03)} & \res{85.85}{(0.35)} & 78.67
			 & \res{55.29}{(4.14)} & \res{62.25}{(1.19)} & \res{64.58}{(1.97)} & 60.71
			 & \res{57.68}{(1.88)} & \res{70.39}{(0.27)} & \res{75.94}{(1.32)} & 68.00
			 & \res{71.50}{(2.90)} & \res{78.63}{(1.08)} & \res{81.35}{(0.84)} & 77.16
			\\
					

			\midrule
			{MT + Contrast (\emph{Enc}+\emph{Dec})} 
			& \underline{\res{86.37} {(0.83)}} & \underline{\res{89.57} {(0.40)}} & \underline{\res{90.40} {(0.15)}} & \underline{88.78}
			& \underline{\res{75.53}{2.62}} & \underline{\res{78.42}{0.40}} & \underline{\res{77.22}{0.74}} & \underline{77.06}
			& \underline{\res{76.11}{1.72}} & \underline{\res{80.21}{0.85}} & \underline{\res{82.00}{1.19}} & \underline{79.44} 
			& 76.16 & 82.89 & 84.85 & 81.30\\
			
			{MT + Ours (\emph{pre-train})} & \textbf{\res{90.25} {(0.38)}} & \textbf{\res{91.36} {(0.18)}} & \textbf{\res{91.04} {(1.16)}} & \textbf{90.88}
			& \textbf{\res{80.16}{0.63}} & \textbf{\res{81.50}{0.66}} & \textbf{\res{78.97}{1.38}} & \textbf{80.21}
			& \textbf{\res{78.71}{1.78}} & \textbf{\res{83.33}{1.11}} & \textbf{\res{83.61}{0.96}} & \textbf{81.88} 
			& \textbf{85.64} & \textbf{85.60} & \textbf{88.45} & \textbf{86.56}\\
			\bottomrule
					 
		\end{tabular}
	\end{adjustbox}
\end{table*}
\subsection{Impact of $\alpha$ in our proposed $\loss{MI}$ objective}

To evaluate the impact of the proposed pseudo-label for the joint distribution matrix, we vary different $\alpha$ based on one of our best performing case and report the results in Table \ref{tab:alpha_ablation}. It can be seen that increasing $\alpha$ from 0 to 0.25 introduces large improvements for all segmentation tasks, which confirms the poor guidance of the IIC objective. Interestingly, we notice that $\alpha=1$ does not lead to the best performance. This might be because clustering pixels into $K=40$ regions of similar sizes breaks the anatomical structures of a given image, and a relatively lower $\alpha$ provides a softer guidance that helps preserve these structures. We confirm this by visualizing in Fig. \ref{fig:alpha_explain} the joint matrix $\Pjoint$, the cluster assignment, as well as the uncertainty of these clusters for different $\alpha$.

\subsection{Impact of boundary-aware loss $\loss{CC}$}
Our boundary-aware loss $\loss{CC}$ is a key component to boost performance for harder segmentation tasks such as RV and Myo. To determine the usefulness of this loss, similar to the previous experiment, we vary $\lambda_{\mr{CC}}$ ranging from 0.0 to 4.0 for our best performing case, and present the results in Table \ref{tab:impact_boundary}. Clearly, increasing $\lambda_{\mr{CC}}$ from 0.0 to 1.0 improves the segmentation performance for all tasks, in particular for the RV and Myo classes whose boundary mainly follows the image edges. In Fig. \ref{fig:boundary}, we show the cluster boundaries separate images obtained with the different $\lambda_{\mr{CC}}$. The boundary-aware loss successfully guides the cluster boundaries towards image edges and reduces the over-segmentation of pixels around the boundaries. 


\subsection{Impact of over-segmented clustering numbers $K$}
Our MI-based method converts continuous feature vectors to a discrete distribution over clusters. The cluster number $K$ is another important hyper-parameter for our method. In this ablation experiment, we measure the impact of $K$ by varying it from 5 to 60. Table \ref{tab:impact_k} and Fig. \ref{fig:impact_k} in Appendix \ref{sec:clusters} show the DSC performance and the corresponding cluster assignment obtained with unsupervised pre-training.  $K=5$ leads to a weak segmentation performance, which can be explained by a collapsed cluster assignment. The cluster maps become more balanced with the increase of $K$ and gradually reflect the cardiac structures in the image. However, using $K\!=\!60$ does not give a better performance in segmentation tasks, since the resulting clusters over-segment the image and capture less relevant regions.  

\subsection{Comparison with state-of-the art methods}

We compare our method with other approaches that aim to improve training with few annotated images/scans. Table \ref{tab:compare_sota} presents results for various semi-supervised learning approaches. For a more detailed explanation of the experimental setup of each method, see Appendix \ref{sec:comparison_methods}. 
To have a fair comparison, for all methods, we used grid search on the validation set to tune the hyper-parameters and report the corresponding test performance. For most methods, the improvement with respect to the baseline trained with only the supervised loss is quite limited and varies depending on the segmentation task and the number of annotated scans used. 
Among different methods, MT offers stable improvements across all tasks and reaches competitive performance compared with contrast (\emph{Dec+Dec}) and Ours (\emph{pre-train}) for the LV and Prostate classes, mainly due to its temporally-ensembled teacher network which provides stable prediction proposals for unlabeled images. However, semi-supervised methods such as MT are normally trained with randomly initialized parameters and can thus be further improved with our proposed \emph{pre-train} approach as initialization of the network. To test this idea, we ran MT with two different initializations, one from contrastive-based Contrast (\emph{Enc+Dec}) and the other from Ours (\emph{pre-train}). The results in the last two rows of Table \ref{tab:compare_sota} indicate that a further improved segmentation is obtained by simply initializing the network parameters with these pre-trained checkpoints: Contrast (\emph{Enc+Dec}) boosted the performance of MT by 2.82\%, 11.20\% and 7.68\% for LV, RV, and Myo, while these improvements increase to 4.93\%, 14.35\% and 10.12\% when Ours (\emph{pre-train}) is used as initialization. In summary, our boundary-aware algorithm for unsupervised representation learning can boost state-of-the-art semi-supervised segmentation approaches to achieve excellent segmentation quality, even when only an extremely small amount of labeled examples are available.

\section{Discussion and conclusion}
In this paper, we presented a boundary-aware information maximization method for the unsupervised pre-training of models for medical image segmentation. This method complements the \emph{global} contrastive loss and can highly improve the performance of a segmentation network when annotated data is scarce. It was shown that, with a reduced amount of unlabeled images, our method can learn a useful local representation on dense feature maps during pre-training, without any supervisory signal. Furthermore, as shown in our visualization of results, the clusters obtained by the pre-trained checkpoint enhance intepretability. We compared our method with recent self-supervised learning approaches, based on clustering and contrastive learning, as well as six strong semi-supervised segmentation algorithms. Results on two benchmark datasets demonstrate the outstanding accuracy of our method. In particular, the combination of our method with Mean Teacher yields unprecedented performance, reaching close to full supervision with a single scan. 

\nocite{langley00}

\bibliography{bibo}
\bibliographystyle{icml2022}

\newpage
\appendix
\onecolumn

\section{Diagram}

\begin{figure}[H]
 \centering
 \includegraphics[width=1\linewidth]{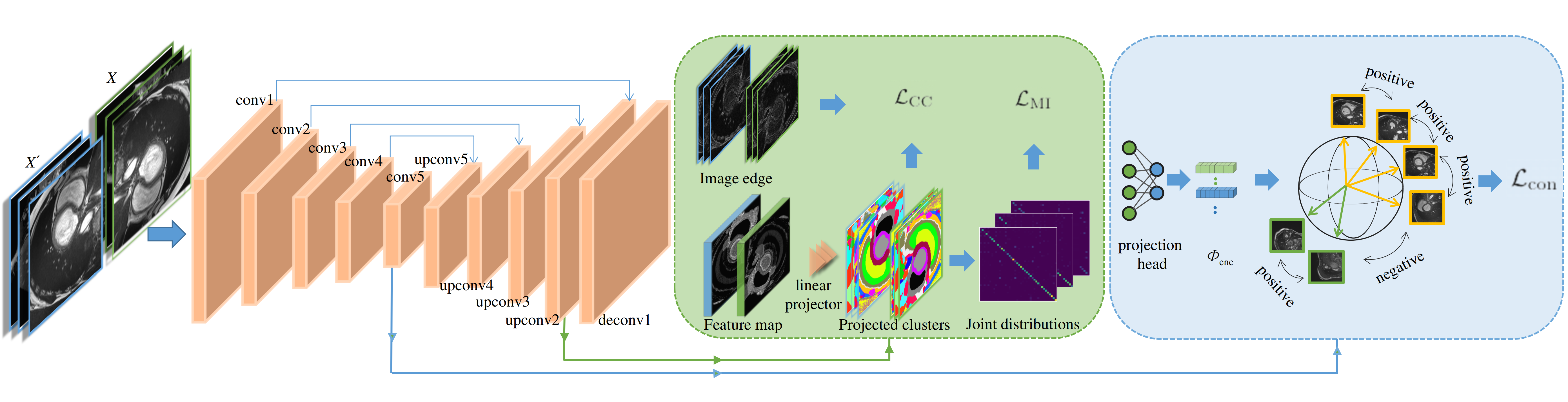}
 \caption{Schematic Diagram of our proposed method. Our method consists of three individual objectives. \emph{Global} contrastive loss $\loss{con}$ enforces images with similar anatomical structures to be pull close, providing global context information for the encoder. Our MI-based loss $\loss{MI}$ aims to cluster the dense local embeddings into $K$ balanced and confident classes, while our boundary-aware loss $\loss{CC}$ aligns the boundaries of these clusters to image edges.}
 \label{fig:conceptual}
\end{figure}

\section{Related work}

Inspired by the recent success of representation learning, various approaches based on pre-training  have been investigated for segmentation. These approaches
seek to acquire a discriminative image representation from unlabeled data, in an independent pre-train stage.
In medical image segmentation, 
\citet{chen2019self, bai2019self} proposed to predict the relative position of patches in MR images. \citet{taleb2021multimodal} extended the jigsaw puzzle solving pretext task to images from multiple MRI modalities. \citet{zhou2021models} proposed Models Genesis, a denoising auto-encoder that reconstructs an MR image given its degraded version as input. Contrastive learning has also shown promising results to boost the performance of downstream tasks using unlabeled images. In this approach, a network is pre-trained to bring closer the feature embeddings of an image under different transformations (positive pairs), while pushing away those from different images (negative pairs). This idea was used to learn a global representation at the end of the network's encoder, using meta-labels on anatomical similarity or subject ID to define the positive pairs \citep{chaitanya2020contrastive, peng2021self, zeng2021positional}. To also pre-train the decoder, the method in \citep{chaitanya2020contrastive} defined positive or negative embedding pairs based on their spatial distance in a feature map, those with a large distance considered as negative while those at the same spatial position but coming from different transformations as positives. 
\citet{hu2021semi} proposed using small set of pixel-wise annotations to guide the learning of dense features in pre-training. The feature embeddings of pixels with the same label are considered as positive pairs and are then clustered together by the contrastive loss. While this guided approach helps learn a better local representation, it requires manual annotations and therefore it is not unsupervised. 
\citet{ouyang2020self} instead employed superpixels for the contrastive objective, however their approach is defined in the context of few-shot segmentation. 

Clustering has also been used to pre-train a network with unlabeled images~\citep{caron2020unsupervised, ji2019invariant, cho2021picie,fang2021self}. Surprisingly, only a few papers have explored this self-supervised learning approach for medical image segmentation \citep{peng2021boosting, ahn2021spatial}. Our method extends the IIC deep clustering approach  ~\citep{ji2019invariant} with an improved loss that encourages clusters to be consistent across different transformations and follow the region boundaries in the image. 

Accurately predicting the boundaries of anatomical structures, tissues or lesions is essential for medical image segmentation, and boundary-aware training methods have been widely explored. To achieve this goal, most methods use a multi-task learning strategy with a secondary loss function focusing on boundary information \citep{li2020shape, xue2020shape, shen2017boundary}. Another approach adopts a discriminator to embed the ground-truth boundary \citep{wei2020attentive}. Unlike our boundary-aware information maximization method for unsupervised representation learning, these approaches require annotated images and thus are limited to supervised or semi-supervised settings. 


\section{Proof of Proposition \ref{the-proposition}}\label{sec:proof}
\begin{proof}
For the lower bound, we use the inequality $\entropy(P,Q) = \KL(P \,\|\, Q) + \entropy(P)$ to obtain
\begin{align}
\entropy(\tfrac{1}{K}\ident_K,\Pjoint) & \, = \, \underbrace{\KL(\tfrac{1}{K}\ident_K \,\|\, \Pjoint)}_{\geq 0} \, + \, \entropy(\tfrac{1}{K}\ident_K)\\
& \, \geq \, -\sum_{j=1}^K\sum_{k=1}^K \tfrac{1}{K}\vone[j=k] \log \big(\tfrac{1}{K}\vone[j=k]\big) \\
& \, = \, -\sum_{k=1}^K \tfrac{1}{K}\log \tfrac{1}{K} \, = \, \log K.
\end{align}
For the upper-bound, we use the Jensen inequality to get
\begin{align}
\entropy(\tfrac{1}{K}\ident_K,\Pjoint) & \, = \, -
\sum_{j=1}^K\sum_{k=1}^K \tfrac{1}{K}\vone[j=k] \log \Pjoint^{(j,k)} \\
& \, = \, 
- \sum_{k=1}^K \tfrac{1}{K} \log \bigg(
\frac{1}{N}\sum_{i=1}^N \widehat{p}_{ik} \, \widetilde{p}_{ik} \bigg) \\
& \, \leq \, 
-\frac{1}{N}\sum_{i=1}^N \sum_{k=1}^K 
\tfrac{1}{K} \big(\log \widehat{p}_{ik} + \log \widetilde{p}_{ik}\big) \\
& \, = \, \frac{1}{N}\sum_{i=1}^N \entropy(\uu, \pHat_i) + \ \entropy(\uu, \pTilde_i) \ \quad\square
\end{align}
\end{proof}

\section{Datasets}
\label{sec:append_dataset}

We assess the performance of our proposed method and compare it with other SOTA approaches on two clinically-relevant datasets: the automatic cardiac diagnosis challenge (ACDC) and the Prostate MR image segmentation 2012 challenge (P\textsc{romise}12). These two datasets cover different anatomical structures, present different acquisition resolutions, and are widely used to verify the effectiveness of semi-supervised segmentation algorithms. 

\textbf{ACDC dataset:
} The ACDC dataset\footnote{Publicly-available by \url{https://www.creatis.insa-lyon.fr/Challenge/acdc/index.html}} consists of 200 short-axis cine-MRI scans from 100 patients, evenly distributed in 5 subgroups: normal, myocardial infarction, dilated cardiomyopathy, hypertrophic cardiomyopathy, and abnormal right ventricles. For each patient, two annotated scans correspond to end-diastolic (ED) and end-systolic (ES) phases are provided, which were acquired on 1.5T
and 3T systems with resolutions ranging from $0.70\ \times\ 0.70$ mm to $1.92\ \times\ 1.92$ mm in-plane and 5 mm to 10 mm through-plane. Three regions of interest: left ventricle endocardium (LV), left ventricle myocardium (Myo), right ventricle endocardium (RV) are labeled from background and delineated pixel-wisely by human experts. 
We consider the 3D-MRI scans as 2D images through-plane due to
the high anisotropic acquisition resolution, and re-sample them to a fix space ranging of $1.0\ \times\ 1.0$ mm. Following \citep{peng2021self}, we normalize the pixel intensities based on the 1\% and 99\% percentile of the
intensity histogram for each scan. Normalized slices are then cropped to $384\ \times\ 384$ pixels, coarsely centered based on the foreground delineation of the ground truth. 
We select slices from 175 random scans as our training set, from which we again randomly select 1, 2 or 4 scans\footnote{These labeled splits were also kept untouched for experiments employing \emph{pre-train} and \emph{fine-tune} strategies, as well as those relying on semi-supervised losses.} as our labeled data, representing representing 0.5\% to 2\% of all available data, and consider others as unlabeled data. We then randomly divide the remaining 25 scans into a validation set and a test set, comprised of 8 scans and 17 scans respectively. Both the validation and test sets were set aside during model optimization. For experiments with \emph{pre-train} and \emph{fine-tune} strategies, we use all \emph{training} data without any pixel-wise annotation for \emph{pre-train} and evaluate the representation ability by fine-tuning the obtained network on a few labeled scans via a cross-entropy loss. We also used various data augmentation transformations $\mathcal{T}(\cdot)$ for both labeled and unlabeled images, including random crops of 224 × 224 pixels, random flip, random rotation, and color jitter from the \texttt{Pillow} library. It is also worthy to notice that the main experimental results from this dataset were obtained by considering three-class segmentation as three binary segmentation tasks. 

\textbf{P\textsc{romise}12 dataset:} Our second dataset\footnote{Publicly-available by \url{https://promise12.grand-challenge.org/}} focuses on prostate segmentation and is composed of multi-centric transversal T2-weighted MR images from 50 subjects. 
These images were acquired from different vendors and with various acquisition protocols, and are thus representative of typical MR images acquired in a clinical setting. Image resolution ranges from 15 × 256 × 256 voxels to $54 \ \times\ 512\ \times\ 512$ voxels with a spacing ranging from $2\ \times\ 0.27\ \times\ 0.27$ mm to $4\ \times \ 0.75\ \times \ 0.75$ mm. Also in this case, we slice these volumetric images into 2D images along the short-axis and resized them to a resolution of $256\ \times\ 256$ pixels. We equally performed a normalization on pixel intensities based on 1\% and 99\% percentile for each scan. We randomly selected 40 scans as
training data, 3 scans for validation, and 7 scans for testing. To test methods in annotation-scarcity regime, we chose 4, 6, and 8 scans from these training examples as the labeled images, while keeping others as unlabeled. We also employed rich data transformation prior to $\mathcal{T}(\cdot)$. These \texttt{Pillow}-based transformations include random crop of $224\ \times\ 224$ pixels, random flip, random rotation within a range of [$-10^{\circ}$, $10^{\circ}$], and color jitter. 

\section{Detailed comparison methods}
\label{sec:comparison_methods}
We implement various state-of-the-art methods for comparison, including: 

 \textbf{Contrastive-based method \citep{chaitanya2020contrastive}:} As our closest work, this approach acquires discriminative representation from unlabeled images via two \emph{global} and \emph{local} contrastive learning. 
 \emph{Global} contrastive learning pre-trains the encoder of the segmentation network to distinguish the global context information such as the anatomical similarities between two slices, 
 while \emph{local} contrastive learning focuses on dense embeddings and enforces pixels undergoing different transformations to be close and pixels at different spatial locations to be pushed away. 
 We refer \emph{Contrast (Enc)} as our \texttt{PyTorch} re-implementation of the variant using only \emph{global} contrastive objective, \emph{Contrast (Dec)} as the variant employing only \emph{local} contrastive learning, and \emph{Contrast (Enc+Dec)} as the full implementation using \emph{both} contrastive objectives. For \emph{Contrast (Enc)} variants, we pre-train the encoder of the segmentation network to distinguish whether two slices comes from similar slice position by assuming the volumetric scans are coarsely aligned. Towards this goal, we manually split an ACDC scans into three partitions, while we fixed the partition numbers for P\textsc{romise}12 as 5, similar to \citet{peng2021self}. A nonlinear projector is used to convert the global representation to representation vector, comprised of an average pooling layer, 2-layer MLP with LeakyReLU as the activation, and a normalization layer. For the variants using \emph{local} contrastive learning, we take the dense embeddings from the layer before last $1\times1$ convolutions. These dense embeddings are then projected to pxiel-wise vectors by a dense projector, consisting of an adaptive average pooling of size $20\times 20$ to reduce the spatial size, 2-layer MLP with $1\times1$ convolutions with LeakyReLU as the activation, and a normalization layer. Positive and negative pairs are then defined on these $20\times20$ grid, similar to \citet{chaitanya2020contrastive}. 
 It is worthy to point out that we only optimize the parameters for the encoder in the \emph{pre-train} stage only when \emph{global} contrastive loss is used, while the whole network except last $1\times1$ convolution is optimized when \emph{local} contrastive loss is employed.
 We employ two stage strategy: \emph{pre-train} and \emph{fine-tune} to evaluate the quality of the learned representation. 
 
 \textbf{IIC (Dec) \citep{ji2019invariant}:} This is the original method proposed in \citep{ji2019invariant} for unsupervised image clustering, which corresponds to our optimization objective $\loss{MI}$ with $\alpha=0.0$. This method has being successfully used in image clustering, as well as for unsupervised natural image segmentation, but it can only work well with coarse classes. We follow the exact protocol and hyper-parameters as our proposed method and report the 3D DSC score on the test set. 
 
 \textbf{IMSAT (Dec) \citep{hu2017learning}:} This method seeks to maximize MI over categorical distributions from dense embeddings in a similar but different formulation: $\MI = \MI(X, p(X))$, where $X$ is the image set while $p(X)$ is the cluster assignment distribution given $X$. This objective has been successfully applied in image clustering \citep{hu2017learning}. We adapt this loss for dense embedding clustering and this method shares the same experimental protocols as our proposed method. We equally evaluate its performance using \emph{pre-train} and \emph{fine-tune} strategy. 
 
 \textbf{Entropy Minimization (EM)~\citep{vu2019advent}: } This method has been successfully applied in semi-supervised classification and segmentation with domain gap, and imposes a low conditional entropy on unlabeled images:
$ \loss{ent} = -\frac{1}{|\data_u||\Omega|} \sum_{x\in \data_u} \sum_{i\in \Omega} p_i(x)\log(p_i(x))$.
 By increasing its confidence for unlabeled images, the network pushes the decision boundary away
 from dense regions of the input space, therefore improving generalization. For this method, we
 performed a hyper-parameter search on the coefficient balancing the cross-entropy and $\loss{ent}$,
 from $1 \times 10^{-4}$ to 1.0. We evaluate this method in a standard semi-supervised setting with randomly initialized network parameters.

 \textbf{MixUp~\citep{zhang2017mixup}:} We also evaluated the effectiveness of mixup, an effective data argumentation strategy
on medical image segmentation, following \citet{chaitanya2020contrastive}. In this method, we interpolate two labeled images with an index sampled from $Beta$($\alpha, \alpha$) distribution and enforce the network to output the prediction as the interpolation of the two annotations. We fix $\alpha=1$ and the coefficient weighting the mixup loss is selected by grid search from $1\times10^{-5}$ to 0.1.

 \textbf{Mean Teacher (MT) \citep{perone2018deep}:} 
 This semi-supervised segmentation method adopts a teacher-student framework, in which two networks sharing the same architecture learn from each other. Given an unlabeled image x, the student model $p^{s}(\cdot)$ seeks to minimize the prediction difference with the teacher network $p^{t}(\cdot)$, whose weights are a temporal exponential moving average (EMA) of the student’s: 
 $\loss{MT} = -\frac{1}{|\data_u||\Omega|} \sum_{x\in \data_u} \sum_{i\in \Omega} |p^{s}_i(x)-p^{t}_i(x)|^2 $.
 We fix the decay coefficient to 0.99. The coefficient balancing the supervised and regularization losses is selected by grid search, from $1 \times 10^{-4}$ to 10.
 
 \textbf{Uncertainty-aware Mean Teacher (UA-MT)~\citep{yu2019uncertainty}:} This semi-supervised approach introduces uncertainty for teacher network, which is achieved by Monte-Carlo dropout through multiple inferences. In our implementation, we forward through the teacher network unlabeled images four times and the uncertainty is obtained by computing the pixel-wise entropy of these predictions. We then use a linearly increased threshold $T$, ranging from $\frac{3}{4}\times \log(K)$ to $\log(K)$ to exclude from $\loss{MT}$ pixels having high uncertainty. We keep other settings the same as our Mean Teacher method and evaluate method's performance in a standard semi-supervised setting.
 
 \textbf{Interpolation Consistency Training (ICT) \citep{verma2019interpolation}:} The next method we tested applies mixup method with teacher-student framework. In this approach, interpolated images are obtained by mixing up two unlabeled images. The student network is encouraged to output the prediction as the interpolation of their predictions given by the teacher network. We follow \citep{verma2019interpolation} to set $\alpha$ as 0.1 and again grid search the weighting coefficient for the regularization objective, from $1\times10^{-5}$ to 0.1.

 \textbf{Adversarial training(AT)~\citep{zhang2017deep}:} Our last method trains a segmentation network and a classifier-based discriminator jointly in a min-max game. The core idea is to enforce
the segmentation predictions on unlabeled images being indistinguishable from those of labeled
images, thus aligning the output distributions between labeled and unseen images. This method works particularly well in a scenario where the image scans present large variability causing a domain gap. We evaluate this method in a standard semi-supervised setting and grid-search the regularization coefficient, from $1\times10 ^{-6}$ to 0.1.

\section{Implementation details}
\label{sec:implementation_details_append}
\textbf{Network Architecture:} We used U-Net \citep{ronneberger2015unet} as our main network architecture, which consists of five symmetric blocks of encoder and decoder. As shown in Fig. \ref{fig:conceptual}, we assign different names to these blocks and our global embeddings and dense embeddings are taken from \texttt{conv5} and \texttt{upconv2}. A first nonlinear projector is used to convert the global representation to representation vector, comprised of an average pooling layer, 2 MLP layers with LeakyReLU as the activation, followed by a normalization layer. In contrast, we simply employ a linear projector, including $1\times1$ convolution followed by a $K$-way softmax for the dense embeddings. Learnable parameters are optimized using stochastic gradient descent (SGD) with a RAdam Optimizer~\citep{liu2019variance}. 

\textbf{Training hyper-parameters:} Our main experiments adopt the two-stage training strategy: pre-training the whole network on all training data without labels and fine-tune it with a few labeled scan. 
For both stages, we employed a learning rate decay strategy, where the initial learning rate $lr$ is increased $N$ times in the first 10 epochs, followed by a cosine decay strategy for the rest $N_{\text{epoch}}$ training epochs. We set $lr = 5\times 10^{-7}$, $N=400$, and $N_{\text{epoch}}=50$ for the ACDC in pre-train stage, $lr=1\times10^{-7}$, $N=200$, and $N_{\text{epoch}}=50$ for ACDC in fine-tune stage. As for the P\textsc{romise}12 dataset, we simply modify $lr$ to $1\times 10^{-6}$ for the fine-tune stage. We define an epoch in our experiments as the $N$ update iterations, within which images are randomly selected from their respective dataset with replacement. For ACDC, we fixed $N$ as 200 iterations while for P\textsc{romise}12, we increase $N$ to 400. For concurrent methods employing semi-supervised setting, we follow exactly the same configuration as adopted in fine-tune stage. 
As shown in Equ. (\ref{equ:total}), our method requires only one weighting coefficient which balances the importance of our $\loss{CC}$ and we simply set it to 1.0 for both datasets.

\textbf{Details on the transformation $\mathcal{T}(\cdot)$:} Our proposed method heavily relies on $\mathcal{T}(\cdot)$ to create transformation equivalent pairs of cluster distribution: $\pHat = \proj(s(\tsf(\img)))$ and $\pTilde = \proj(\tsf(s(\img)))$. We set $\mathcal{T}(\cdot)$ as the cascade of intensity transformations and geometric transformations. When $\mathcal{T}(\cdot)$ takes an input $\img$ as the raw image, we apply gamma correction within a range of [0.5, 2.0], as well as a set of random affine transformation, consisting of random scale within a range of [0.8, 1.3], random rotation within a range of [$-45^{\circ}$, $45^{\circ}$], and random translation within a range of [$-10\%$, $10\%$]. Whereas when $\mathcal{T}(\cdot)$ takes the input as the embedding $s$ of the image $\img$, we ignore the intensity transformation and apply only the random affine transformation with the same random state corresponding to those applied with the raw image $\img$. These augmentations operate on \texttt{PyTorch} tensors and are publicly-available at \url{https://github.com/PhoenixDL/rising.git}

\section{Pre-trained cluster assignment maps for different $K$}\label{sec:clusters}
We show in Fig. \ref{fig:impact_k} the pre-trained cluster assignment for different number of clusters $K$. One can see that a small $K$ learns a collapsed cluster assignment, which leads to a weak segmentation performance (see Table \ref{tab:impact_k}). This is probably because a small cluster number reduces the capacity to capture the structure information of such images. With the increase of $K$, the cluster maps become more balanced and gradually reflect the cardiac structures of the image. However, when taking a large cluster numbers $K=60$, the resulted clusters over-segment the images, leading to fractured anatomical structures. In this case, it can decrease the downstream segmentation tasks.

\begin{figure}[H]
	\centering
	\begin{adjustbox}{max width=0.6\linewidth}
		\renewcommand{\arraystretch}{1}
		\setlength{\tabcolsep}{1pt}
		\begin{small}
			\begin{tabular}{cccccc}
				\includegraphics[width=0.166\linewidth]{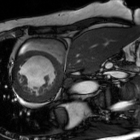} &
				\includegraphics[width=0.166\linewidth]{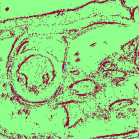}& 
				\includegraphics[width=0.166\linewidth]{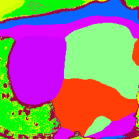}&
				\includegraphics[width=0.166\linewidth]{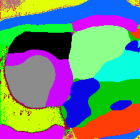} & 
				\includegraphics[width=0.166\linewidth]{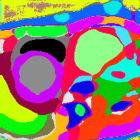}& 
				\includegraphics[width=0.166\linewidth]{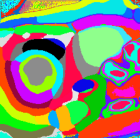} \\
				\includegraphics[width=0.166\linewidth]{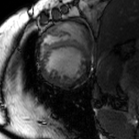} &
				\includegraphics[width=0.166\linewidth]{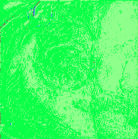}& 
				\includegraphics[width=0.166\linewidth]{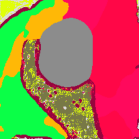}&
				\includegraphics[width=0.166\linewidth]{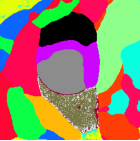} & 
				\includegraphics[width=0.166\linewidth]{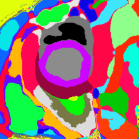}& 
				\includegraphics[width=0.166\linewidth]{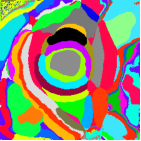} \\
				\includegraphics[width=0.166\linewidth]{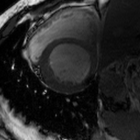} &
				\includegraphics[width=0.166\linewidth]{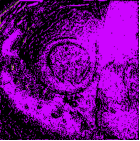}& 
				\includegraphics[width=0.166\linewidth]{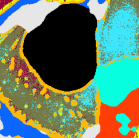}&
				\includegraphics[width=0.166\linewidth]{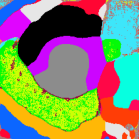} & 
				\includegraphics[width=0.166\linewidth]{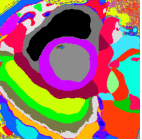}& 
				\includegraphics[width=0.166\linewidth]{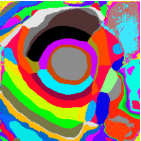} \\
				Image 
				& \tabincell{c}{$K\!=\!5$} 
				& \tabincell{c}{$K\!=\!10$} 
				& \tabincell{c}{$K\!=\!20$}
				& \tabincell{c}{$K\!=\!40$}
				& \tabincell{c}{$K\!=\!60$}
			\end{tabular}
		\end{small}
	\end{adjustbox}
	\caption{Pre-trained cluster assignment with respect to different $K$}
	\label{fig:impact_k}
\end{figure}

\section{Impact of batch size on pre-training}

As contrastive-based pre-training often requires a large batch size, which is hard to satisfy for dense prediction tasks, such as segmentation. Our last ablation study investigates the performance stability given relatively small batch size $\mathcal{B}$. Table \ref{tab:batch_size} lists the 3D test DSC for ACDC dataset with reduced batch size for contrastive-based and one of our best performing variant.  It can be seen that with reduced batch size, segmentation performances reduces for both methods. However, our proposed method still outperforms contrastive-based approach for almost all cases given a very small batch.

\begin{table}[h!]
	\centering
	\caption{Impact of batch size $\mathcal{B}$}
	\label{tab:batch_size}
	\begin{adjustbox}{max width=0.55\linewidth}
		\setlength{\tabcolsep}{3pt}
		\renewcommand{\arraystretch}{1}
		\begin{tabular}{c|ccc|ccc|ccc}
			\toprule
			\multirow[b]{2}{*}{\phz$\mathcal{B}$\phz} 
			& \multicolumn{3}{c|}{\textbf{ACDC-LV}}
			& \multicolumn{3}{c|}{\textbf{ACDC-RV}}
			& \multicolumn{3}{c}{\textbf{ACDC-Myo}}\\
			\cmidrule(l{0pt}r{0pt}){2-4} 
			\cmidrule(l{0pt}r{0pt}){5-7} 
			\cmidrule(l{0pt}r{0pt}){8-10} 
			 & \bfseries 1 scan & \bfseries 2 scans & \bfseries 4 scans & \bfseries 1 scan & \bfseries 2 scans & \bfseries 4 scans & \bfseries 1 scans & \bfseries 2 scans & \bfseries 4 scans \\
			\midrule
			 & \multicolumn{9}{c}{Contrast (\emph{Enc+Dec})} \\
			\midrule
			\phz6 & 72.16 & 86.02 & 87.37 
			& 63.30 & 69.25 & 72.25 
			& 61.51 & 73.57 & 76.51 \\
            12 & 75.52 & 84.23 & 88.31 
            & 63.26 & 71.41 & 73.89 
            & 65.88 & 76.26 & 78.77 \\
            18 & 77.98 &85.97 &88.42 
            & 66.47 & 72.82& 76.69 
            & 64.96 &76.98& 78.76\\
			\midrule
			& \multicolumn{9}{c}{Ours (\emph{MI+CC})} \\
			\midrule
			\phz6 & 81.61 & 85.76 & 88.21 
			& 67.39 & 67.04 & 66.08 
			& 71.18 & 77.41 & 80.20 \\
			12 & 81.46 & 87.89 & 88.72 
			& 68.15 & 76.33 & 74.96 
			& 74.84 & 78.54 & 82.58 \\
			18 & 84.04 & 88.52 & 89.31 
			& 76.86 & 79.13 & 75.92 
			& 76.93 & 79.59 & 81.97 \\
			\bottomrule
		\end{tabular}
	\end{adjustbox}
	 
\end{table}

\newpage
\section{Visual results for segmentation}
\begin{figure}[h!]
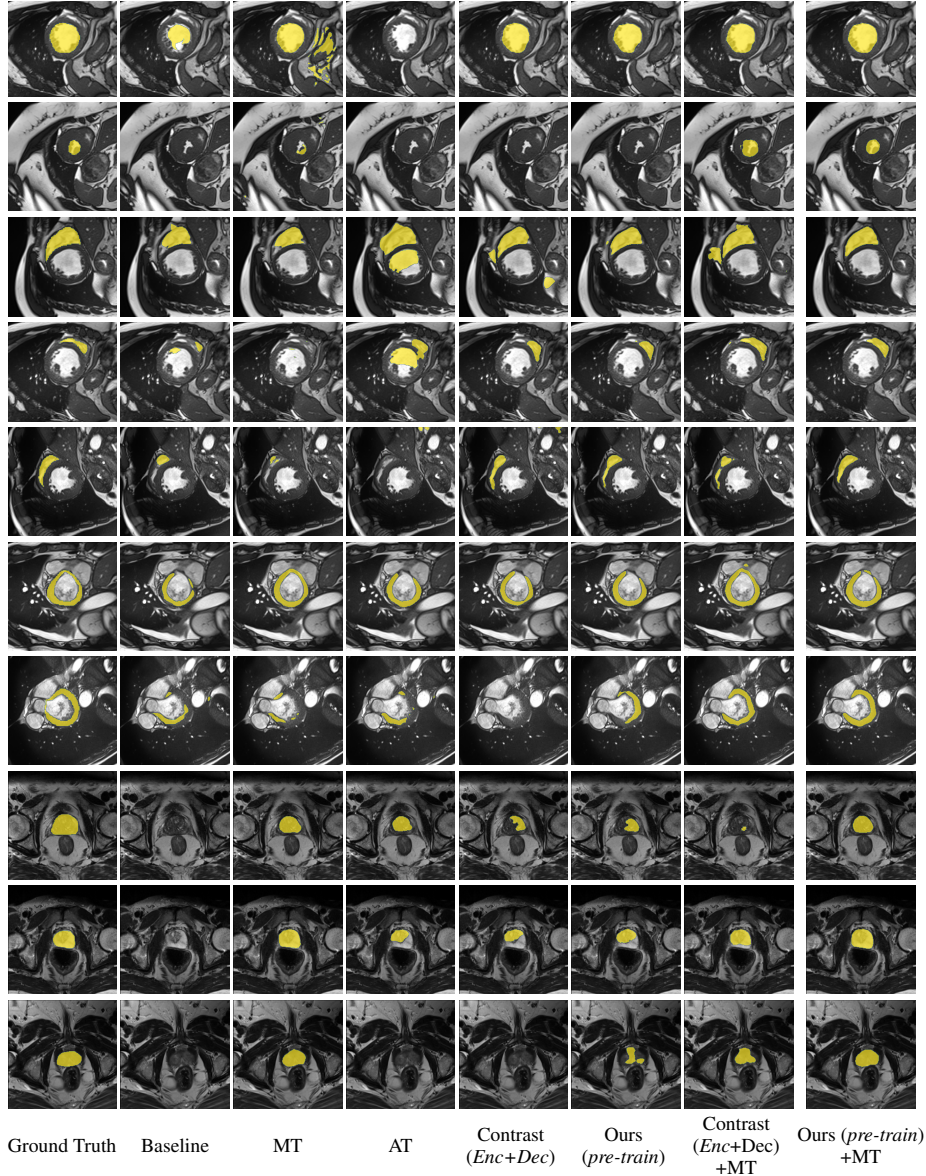

    \centering
    \begin{adjustbox}{max width=0.9\linewidth}
        \begin{footnotesize}
		\setlength{\tabcolsep}{1pt}
		\renewcommand{\arraystretch}{1}
        \begin{tabular}{cccccccc}
         \visinpsect{LV}{1}  \\
          \visinpsect{LV}{2}  \\
        \visinpsect{RV}{1}  \\
      \visinpsect{RV}{3}  \\
        \visinpsect{RV}{4}  \\
        \visinpsect{MYO}{1}  \\
        \visinpsect{MYO}{2}  \\
        \visinpsect{prostate}{1}  \\
        \visinpsect{prostate}{2}  \\
        \visinpsect{prostate}{3}  \\
        Ground Truth & Baseline & MT & AT & \tabincell{c}{Contrast\\(\emph{Enc+Dec})} & \tabincell{c}{Ours \\(\emph{pre-train})} & {\tabincell{c}{Contrast \\(\emph{Enc}+{Dec})\\+MT}}  & \tabincell{c}{Ours (\emph{pre-train})\\ +MT}
        
        \end{tabular}
        \end{footnotesize}
\end{adjustbox}
    \caption{Visual comparison of tested methods on test images. Rows 1–2: \textbf{LV};
Rows 3–5: \textbf{RV}; Rows 6–7: \textbf{Myo}; Row 8-10: \textbf{P\textsc{romise}12}. }
    \label{fig:visual_inspection}
\end{figure}

\end{document}